\documentclass[10pt]{article}

\newcommand{\arxivversion}[1]{#1}
\newcommand{\confversion}[1]{}
\newcommand{\arxivconf}[2]{\arxivversion{#1}\confversion{#2}}

\newcommand{\hnote}[1]{}

\usepackage{amsfonts}

 \usepackage{amsmath}
\usepackage{mathabx}
\usepackage{comment}

\confversion{
\ifCLASSINFOpdf
\else
\fi
}
\usepackage{amsfonts}

\usepackage{amsmath}
\usepackage{mathabx}
\usepackage{comment}
\usepackage{cancel}

\begin{document}
%
\title{
Counting Answers to
Existential Positive Queries:\\
A Complexity Classification
}

\arxivconf{
\author{
Hubie Chen\\Universidad del Pa\'{i}s Vasco, E-20018 San Sebasti\'{a}n, Spain\\
\emph{and} IKERBASQUE, Basque Foundation for Science, E-48011 Bilbao, Spain
 \and
Stefan Mengel\\CNRS, CRIL UMR 8188, France
}}
{
\author{
\IEEEauthorblockN{Hubie Chen}
\IEEEauthorblockA{Departamento LSI, 
Universidad del Pa\'{i}s Vasco\\
\emph{and} IKERBASQUE, Basque Foundation for Science\\
E-20018 San Sebasti\'{a}n, Spain
}
\and
\IEEEauthorblockN{Stefan Mengel}
\IEEEauthorblockA{
LIX UMR 7161\\
\'{E}cole Polytechnique\\ Universit\'{e} Paris Saclay,
France}
}
}


\date{ }
\maketitle

\begin{abstract}
Existential positive formulas form a fragment of first-order logic that
includes and is 
semantically
equivalent to unions of conjunctive queries, one of the most important and well-studied classes of queries in database theory.
We consider the complexity of counting the number of answers to 
existential positive formulas on finite structures and give a 
trichotomy theorem on query classes, in the setting of bounded arity. 
This theorem generalizes and unifies
 several known results on the complexity of conjunctive queries and unions of conjunctive queries. 
\end{abstract}


%


\confversion{\IEEEpeerreviewmaketitle}

\newtheorem{theorem}{Theorem}[section]
\newtheorem{conjecture}[theorem]{Conjecture}
\newtheorem{corollary}[theorem]{Corollary}
\newtheorem{examples}[theorem]{Examples}
\newtheorem{proposition}[theorem]{Proposition}
\newtheorem{prop}[theorem]{Proposition}
\newtheorem{lemma}[theorem]{Lemma}
\newtheorem{definition}[theorem]{Definition}
\newtheorem{observation}[theorem]{Observation}
\newtheorem{remark}[theorem]{Remark}
\newtheorem{examplecore}[theorem]{Example}

\newenvironment{example}
  {\begin{examplecore}\rm}
  {\hfill $\Box$\end{examplecore}}
%



\newenvironment{proof}{\noindent\textbf{Proof\/}.}{\hfill$\Box$\medskip}

\newcommand{\FPT}{\mathsf{FPT}}

\newcommand{\sh}{\sharp}

\newcommand{\nats}{\mathbb{N}}
\newcommand{\N}{\mathbb{N}}
\newcommand{\Z}{\mathbb{Z}}

\newcommand{\rela}{\mathbf{A}}
\newcommand{\relb}{\mathbf{B}}
\newcommand{\relc}{\mathbf{C}}
\newcommand{\reld}{\mathbf{D}}
\newcommand{\relt}{\mathbf{T}}
\newcommand{\relp}{\mathbf{P}}
\newcommand{\relf}{\mathbf{F}}
\newcommand{\relg}{\mathbf{G}}
\newcommand{\relh}{\mathbf{H}}
\newcommand{\reli}{\mathbf{I}}
\newcommand{\relm}{\mathbf{M}}
\newcommand{\reln}{\mathbf{N}}
\newcommand{\relr}{\mathbf{R}}

\newcommand{\calP}{\mathcal{P}}
\newcommand{\calE}{\mathcal{E}}

\newcommand{\free}{\mathrm{free}}
\newcommand{\closed}{\mathrm{closed}}
\newcommand{\lib}{\mathrm{lib}}

\newcommand{\id}{\mathrm{id}}
\newcommand{\surj}{\mathrm{surj}}
\newcommand{\tw}{\mathsf{tw}}
\newcommand{\starsize}{\mathsf{starsize}}
\newcommand{\qaw}{\mathsf{qaw}}
\newcommand{\topp}{\mathsf{top}}
\newcommand{\atom}{\mathsf{atom}}

\newcommand{\res}{\upharpoonright}
\newcommand{\aug}{\mathsf{aug}}
\newcommand{\width}{\mathsf{width}}
\newcommand{\shwidth}{\sh\textup{-}\width}
\newcommand{\contract}{\mathsf{contract}}

\newcommand{\p}{\mathsf{P}}
\newcommand{\np}{\mathsf{NP}}

\newcommand{\FO}{\mathsf{FO}}
\newcommand{\PP}{\mathsf{PP}}
\newcommand{\EP}{\mathsf{EP}}

\newcommand{\fo}{\mathsf{FO}}
\newcommand{\pp}{\mathsf{PP}}
\newcommand{\ep}{\mathsf{EP}}

\newcommand{\countp}{\mathsf{count}}
\newcommand{\pcountp}{\smallp\countp}

\newcommand{\dom}{\mathrm{dom}}
\newcommand{\powfin}{\wp_{\mathsf{fin}}}

\newcommand{\pn}[1]{\textsc{#1}}

\newcommand{\smallp}{\ensuremath{p\textup{-}}}
\newcommand{\clique}{\pn{Clique}}
\newcommand{\sclique}{\pn{\#Clique}}
\newcommand{\pclique}{\smallp\clique}
\newcommand{\psclique}{\smallp\sclique}

\newcommand{\param}[1]{\mathsf{param}\textup{-}#1}



\section{Introduction}

\subsection{Background}
The computational problem of evaluating a formula (of some logic)
on a finite relational structure is of central interest in 
database theory and logic.
In the context of database theory, this problem is
often referred to as \emph{query evaluation}, as
it models the posing of a query to a database, in a well-acknowledged way: 
the formula is the query,
and the structure represents the database.
We will refer to the results of such an evaluation as
\emph{answers}; 
logically, these are the satisfying assignments of the formula on the structure.
The particular case of this problem where the formula is a sentence 
is often referred to as \emph{model checking},
and even in just the case of first-order sentences,
can capture a variety of well-known decision problems
from all throughout computer science~\cite{FlumGrohe06-parameterizedcomplexity}.

In this article, we study the counting version of
this problem, namely, given a formula and a structure,
output the \emph{number} of answers
(see for example~\cite{PichlerSkritek11-counting,GrecoScarcello14-counting,DurandMengel13-structuralcounting,ChenMengel14-pp-arxiv} for previous studies).
This problem 
of \emph{counting query answers}
generalizes model checking, which can be viewed as the 
particular case thereof where one is given a sentence and structure,
and wants to decide if the number of answers is $1$ or $0$, 
corresponding to whether or not the empty assignment is satisfying.
In addition to the counting problem's basic and fundamental interest, 
it can be pointed out that all practical query languages
supported by database management systems have
a counting operator.
Indeed, it has been argued that database queries with counting
are at the basis of decision support systems 
that handle large data volume~\cite{GrecoScarcello14-counting}.

As has been previously articulated,
a typical situation in the database setting is the evaluation of a relatively short formula on a relatively large structure.  Consequently, it has been argued that, in measuring the time complexity of query evaluation tasks, one could reasonably allow a slow (non-polynomial-time) preprocessing of the formula, so long as the desired evaluation can be performed in polynomial time following the preprocessing~\cite{PapadimitriouYannakakis99-database,FlumGrohe06-parameterizedcomplexity}.
Relaxing polynomial-time computation to allow arbitrary preprocessing of a \emph{parameter} of a problem instance yields, in essence, the notion of \emph{fixed-parameter tractability}.
This notion of tractability is at the core of \emph{parameterized complexity theory}, which provides a taxonomy for classifying problems where each instance has an associated parameter.  We make use of this paradigm in this article; here, the formula is the parameter.

\subsection{Contribution}

\emph{Existential positive queries}
are the first-order formulas built from
the two binary connectives ($\wedge, \vee$)
and existential quantification.  
They include and are semantically equivalent to
the so-called
\emph{unions of conjunctive queries}, also known as
\emph{select-project-join-union queries};
these have been argued to be the most 
common database queries~\cite{AbiteboulHullVianu95-foundationsdatabases}.
Indeed, each union of conjunctive queries 
can be viewed as
an existential positive query having a particular form, namely,
a disjunction of primitive positive formulas;
recall that a \emph{primitive positive} query is an existential positive
query that does not use disjunction.

We study the problem of counting query answers
on existential positive queries.
An established way to understand which types of queries are 
computationally well-behaved and exhibit desirable,
tractable behavior is to consider this problem relative to a set
of queries,
and to attempt to understand on which sets 
this problem is tractable.
Precisely,
each set $\Phi$ of existential positive queries 
yields a restricted version of the general problem,
namely: count the number of answers of a given formula $\phi \in \Phi$
on a given finite structure~$\relb$.
We hence have a family of problems, one problem for each such set $\Phi$.
Our study focuses on formula sets that have
\emph{bounded arity}
(by which is meant that there is a constant that upper bounds the arity of all relation symbols used in formulas);
let us assume this property of all formula sets in this discussion.\footnote
{
Note that in the case of unbounded arity, complexity may depend on the choice of representation of
relations~\cite{ChenGrohe10-succinct}.
}

In this article, 
we  prove a trichotomy theorem 
(Theorem~\ref{thm:trichotomy})
on the parameterized complexity of the
discussed family of problems,
which describes the complexity of every such problem.
In particular, our trichotomy theorem 
shows that---in a sense made precise---each such problem
 is fixed-parameter tractable,
equivalent to the \emph{clique} problem, 
or as hard as the \emph{counting clique} problem (which generalizes the clique problem).
Note that the hypothesis that the clique problem is not fixed-parameter
tractable is an established one
in parameterized complexity;\footnote{It can be phrased in terms of complexity classes: FPT $\neq$ W[1].} 
under this hypothesis, our trichotomy theorem
yields a precise description of the problems (from those under consideration)
that are fixed-parameter tractable.
Our trichotomy theorem is in fact derived by invoking two theorems:
\begin{itemize}

\item A new theorem
showing that,
for each
set of existential positive queries, there exists
a set of primitive positive queries such that the
two sets exhibit the same complexity behavior
(see Theorem~\ref{thm:equivalence-theorem}).
This new theorem,
which we call the \emph{equivalence theorem}, 
can be conceived of as the primary technical
contribution of this article.

\item 
A previously presented trichotomy
on primitive positive queries~\cite{ChenMengel14-pp-arxiv,ChenMengel15-pp-icdt} (discussed in Section~\ref{sct:ppclassification}.)

\end{itemize}


\subsection{Related work}

The statement of our new trichotomy theorem 
generalizes, unifies, and strengthens a number of existing 
parameterized complexity classification results in the literature, namely:
\begin{itemize}

\item The dichotomy for model checking primitive positive formulas~\cite{Grohe07-otherside}, 
which built on a previous dichotomy~\cite{GroheSchwentickSegoufin01-conjunctivequeries}; see also~\cite{ChenMueller14-hierarchy}.

\item The dichotomy for model checking existential positive formulas~\cite{Chen14-existentialpositive}.

\item The dichotomy for counting answers to quantifier-free primitive positive formulas~\cite{DalmauJonsson04-counting} (phrased as the problem of counting homomorphisms between relational structures).

\item The trichotomy for 
counting answers to primitive positive formulas~\cite{ChenMengel14-pp-arxiv,ChenMengel15-pp-icdt}, which trichotomy built on the previous work~\cite{DurandMengel13-structuralcounting}.

\end{itemize}
Let us emphasize that we only claim to generalize the parameterized
complexity versions of the presented results.  In
some of the above works, such as 
the works \cite{Grohe07-otherside}
and~\cite{DalmauJonsson04-counting},
the problems that are
classified as
fixed-parameter tractable are also polynomial-time tractable.
We can further remark that 
there are problems
from the dichotomy theorem 
on model checking existential positive formulas~\cite{Chen14-existentialpositive} that are shown to be fixed-parameter tractable
but also NP-complete.

The techniques used to prove our equivalence theorem 
are algebraic and combinatorial,
and are quite different in nature from and contrast with
those used to prove the previous classifications,
which were more graph-theoretic and logical in flavor.
Indeed, while the graph-theoretic measure of
\emph{treewidth} played a key role in the 
statement and proof of the previous trichotomy
as well as of the previous dichotomies on primitive positive queries,
it is not at all needed to prove our equivalence theorem.
To establish the equivalence theorem, we make a key application
of the inclusion-exclusion counting principle
to translate an existential positive formula to 
a finite set of primitive positive formulas
(see Section~\ref{subsect:all-free}),
which, in the setup considered by the article,
is crucial to handling and understanding disjunction.
We believe that the developed theory that supports 
said application should provide a valuable foundation
for coping with disjunction in logics that are
more expressive than the one considered here.



\section{Preliminaries}

\subsection{Basic definitions and notions}

Note that $\cdot$ is sometimes used for multiplication of 
real
numbers.

{\bf Polynomials.}
We remind the reader of some basic facts about polynomials
 which we will use throughout the paper.  
Here, a univariate polynomial $p$ in a variable $x$ is a 
function $p(x) =\sum_{i=0}^d a_i x^i$ where $d \geq 0$,
each $a_i \in \mathbb{R}$
and
$a_d\ne 0$, or the \emph{zero polynomial} $p(x) = 0$.
The $a_i$ are called \emph{coefficients} of $p$.
The degree of a polynomial is defined as $-\infty$ 
in the case of the zero polynomial, and as $d$ otherwise.
Let $(x_0, y_0), \ldots, (x_n, y_n)$ be $n+1$ pairs of real numbers. Then there is a uniquely determined polynomial of degree at most $n$ such that $p(x_i)=y_i$ for each $i$;
consequently,
a polynomial $p$ of degree $n$ that has at least $n+1$ zeroes 
(where a \emph{zero} is a value $x$ such that $p(x) = 0$)
is the zero polynomial.
If all $x_i$ and $y_i$ are rational numbers, then the coefficients $a_i$ of this polynomial are rational numbers as well;
moreover, the $a_i$ can be computed in polynomial time.

{\bf Logic.}
We assume basic familiarity with the syntax and semantics of first-order logic.
In this article, we focus on relational first-order logic
where equality is not built-in to the logic.
Hence, each \emph{vocabulary/signature} under discussion consists
only of relation symbols.
We assume structures under discussion to be \emph{finite}
(that is, have finite universe);
nonetheless, 
we sometimes describe structures as \emph{finite} for emphasis.
We assume that the relations of structures are represented
as lists of tuples.
We use the letters $\rela$, $\relb$, $\ldots$ to denote structures,
and the corresponding letters $A$, $B$, $\ldots$
to denote their respective universes.
When $\tau$ is a signature, we use $\reli_{\tau}$ 
to denote 
the $\tau$-structure 
with universe $\{ a \}$ and where
each relation symbol $R\in \tau$ has $R^\reli =\{ (a,\ldots, a)\}$. 
When $\rela, \relb$ are structures over the same signature $\tau$,
a \emph{homomorphism} from $\rela$ to $\relb$
is a mapping $h: A \to B$ such that, for each $R \in \tau$
and each tuple $(a_1, \ldots, a_k) \in R^{\rela}$,
it holds that $(h(a_1), \ldots, h(a_k)) \in R^{\relb}$.

We use the term \emph{fo-formula} to refer to a first-order formula.
An \emph{ep-formula} (short for \emph{existential positive formula}) 
is a fo-formula built from
\emph{atoms} 
(by which we refer to predicate applications of the form $R(v_1, \ldots, v_k)$, where $R$ is a relation symbol and the $v_i$ are variables),
conjunction ($\wedge$), disjunction ($\vee$),
and existential quantification ($\exists$).
A \emph{pp-formula}
(short for \emph{primitive positive formula}) 
is defined as an ep-formula where disjunction does not occur.
An fo-formula is \emph{prenex} if it has the form
$Q_1 v_1 \ldots Q_n v_n \theta$ where $\theta$ is quantifier-free,
that is, if all quantifiers occur in the front of the formula.
The set of free variables of a formula $\phi$ is denoted by
$\free(\phi)$ and is defined as usual;
a formula $\phi$ is a \emph{sentence} if $\free(\phi) = \emptyset$.

We now present some definitions and conventions that
are not totally standard.
A primary concern in this article 
is in counting satisfying assignments of fo-formulas
on a finite structure.
The count is sensitive to the set of variables over which
assignments are considered; and, we will sometimes 
(but not always) want to count
relative to a set of variables that is strictly larger than
the set of free variables.
Hence, we will often associate with each fo-formula $\phi$
a set $V$ of variables called the \emph{liberal variables},
denoted by $\lib(\phi)$,
which is required to be~a 
superset of $\free(\phi)$, that is, 
we require $\lib(\phi) \supseteq \free(\phi)$.
Note that $\lib(\phi)$ may contain variables that do not
occur at all in atoms of $\phi$.
To indicate that $V$ is the set of liberal variables of $\phi$,
we often use the notation $\phi(V)$;
we also use $\phi(v_1, \ldots, v_n)$, where the $v_i$ 
are a listing of the elements of $V$.
Relative to a formula $\phi(V)$, when~$\relb$ is a structure,
we will use $\phi(\relb)$ to denote the set of
assignments $f: V \to B$ such that $\relb, f \models \phi$.
\emph{We assume that, in each prenex formula
with
liberal variables associated with it,
no variable is both liberal and quantified.}
We call an fo-formula $\phi$ \emph{free} if $\free(\phi) \neq \emptyset$,
and
\emph{liberal} if 
$\lib(\phi)$ is defined and $\lib(\phi) \neq \emptyset$.

\begin{example}
Let us consider the formula $\phi(x,y,z) = R(x,y) \vee S(y,z)$.
As indicated above, the notation $\phi(x,y,z)$ is used to indicate
that $\lib(\phi) = \{ x, y, z \}$.  
As $\free(\phi) = \{ x, y, z \}$, we have $\lib(\phi) = \free(\phi)$.
Define $\psi(x,y,z) = R(x,y)$ and $\psi'(x,y,z) = S(y,z)$.
By the notation $\psi(x,y,z)$, we indicate that
$\lib(\psi) = \{ x, y, z \}$; likewise, it holds that 
$\lib(\psi') = \{ x, y, z \}$.
Notice that $\free(\psi) = \{ x, y \}$, so we have that
$\lib(\psi)$ is a proper superset of $\free(\psi)$;
in fact, the variable $z \in \lib(\psi)$ does not occur
at all in an atom of $\psi$.
Define also $\theta(x,y) = R(x,y)$; by the notation
$\theta(x,y)$, we indicate that $\lib(\theta) = \{ x, y \}$.

Observe that, for any structure $\relb$, we have
$\phi(\relb) = \psi(\relb) \cup \psi'(\relb)$
(and hence $|\phi(\relb)| = |\psi(\relb) \cup \psi'(\relb)|$).
Observe, however, that for any structure $\relb$
where $\theta(\relb)$ is non-empty, it does not
hold that 
$\phi(\relb) = \theta(\relb) \cup \psi'(\relb)$,
since $\phi(\relb)$ contains only assignments defined on 
$\lib(\phi) = \{ x, y, z \}$, whereas
$\theta(\relb)$ contains only assignments defined on
$\lib(\theta) = \{ x, y \}$.
\end{example}

{\bf pp-formulas.}
It is well-known~\cite{ChandraMerlin77-optimal}
that there is a correspondence between prenex pp-formulas
and relational structures.
In particular, 
each prenex pp-formula $\phi(S)$ (on signature $\tau$) 
with $\lib(\phi) = S$
may be viewed as a
pair $(\rela, S)$ 
consisting of a structure $\rela$ (on $\tau$) 
and a set $S$;
the universe $A$ of $\rela$ is
the union of $S$ with the variables appearing in $\phi$,
and the following condition 
defines the relations of $\rela$: for each $R \in \tau$,
a tuple $(a_1, \ldots, a_k) \in A^k$
is in $R^{\rela}$ 
if and only if $R(a_1, \ldots, a_k)$ appears in $\phi$.
In the other direction, 
such a pair $(\rela, S)$ can be viewed as a prenex pp-formula $\phi(S)$
where all variables in $A \setminus S$ are quantified
and the atoms of $\phi$ are defined according to the above condition.
A basic known fact~\cite{ChandraMerlin77-optimal}
that we will use
is that
when $\phi(S)$ is a pp-formula corresponding to the pair $(\rela, S)$,
$\relb$ is an arbitrary structure, and $f: S \to B$ is an arbitrary map,
it holds that 
$\relb, f \models \phi(S)$ if and only if
there is an extension $f'$ of $f$ that is a homomorphism
from~$\rela$ to $\relb$.
\emph{We will freely interchange between the
 structure view
and the usual notion
of a prenex pp-formula.}
For a pp-formula specified as a 
pair $(\rela,S)$, we typically assume that $S \subseteq A$.

\begin{example}
\label{ex:pp-formula}
Consider the  pp-formula
$\phi(x,x',y,z) = \exists y' \exists u \exists v \exists w (E(x, x') \wedge E(y,y') \wedge F(u,v) \wedge G(u,w)).$
The notation $\phi(x, x',y,z)$ 
indicates that $\lib(\phi) = \{ x, x', y, z \}$.
Note that $\free(\phi) = \{ x, x', y \}$.
To convert $\phi$ to a structure $\rela$, 
we take the universe $A$ of $\rela$ to be
the union of $\lib(\phi)$ with all variables appearing in $\phi$,
so $A = \{ x, x',y,z, y', u, v, w \}$.
Then, we define the relations as just described above,
so $E^{\rela} = \{ (x, x'), (y, y') \}$, 
$F^{\rela} = \{ (u,v) \}$, and
$G^{\rela} = \{ (u,w) \}$.
The resulting pair representation of $\phi$ 
is $(\rela, \{ x, x', y, z \})$.
\end{example}

Two structures are \emph{homomorphically equivalent} if
each has a homomorphism to the other.
A  structure is a \emph{core} if it is not homomorphically equivalent to a proper substructure of itself.  A structure $\relb$ is a 
\emph{core} of a structure $\rela$ if $\relb$ is a substructure of $\rela$ that is a core and is homomorphically equivalent to $\rela$.
It is known that all cores of a structure are isomorphic and hence one sometimes speaks of \emph{the} core of a structure.

For a prenex pp-formula $(\rela,S)$ 
on signature $\tau$,
we define its \emph{augmented structure},
denoted by
$\aug(\rela,S)$,
to be the structure
 over the expanded vocabulary $\tau\cup \{R_a\mid a\in S\}$ (understood to be a disjoint union)
where $R_a^{\aug(\rela,S)} =\{a\}$;
we define the \emph{core} of the pp-formula $(\rela,S)$
to be the core of~$\aug(\rela,S)$.

The following fundamental facts on pp-formulas
 will be used throughout.

\begin{theorem}[follows from \cite{ChandraMerlin77-optimal}]\label{thm:ChandraMerlin}
Suppose
 that each of the pairs $(\rela,V)$, $(\relb,V)$ is a prenex pp-formula.
The formula $(\relb,V)$ logically entails
the formula $(\rela,V)$ if and only if
there exists a homomorphism from the structure
$\aug(\rela,V)$ to the structure $\aug(\relb,V)$.
The formulas $(\rela, V)$, $(\relb,V)$ are
logically equivalent if and only if
they have isomorphic cores, or equivalently, when
$\aug(\rela,V)$ and $\aug(\relb,V)$ are homomorphically equivalent.
\end{theorem}

{\bf ep-formulas.}
In order to discuss ep-formulas, we will employ the following
terminology.
An ep-formula is \emph{disjunctive} 
if it is the disjunction of prenex pp-formulas;
when $\phi$ is a disjunctive ep-formula
with $\lib(\phi)$ defined, we typically assume
that each of the pp-formulas $\psi$ that appear as disjuncts
of $\phi$ has $\lib(\psi)$ defined as $\lib(\phi)$.
(In this way, for an arbitrary finite structure $\relb$,
it holds that $|\phi(\relb)| = |\bigcup_{\psi} \psi(\relb)|$,
where the union is over all such disjuncts $\psi$.)
An ep-formula is \emph{all-free} if it is disjunctive
and each pp-formula appearing as a disjunct is free.
An ep-formula $\phi(S)$ is \emph{normalized} if it is disjunctive
and for each sentence disjunct
$(\rela, S)$ and any other disjunct $(\rela', S)$,
there is no homomorphism from 
$\aug(\rela, S)$ to $\aug(\rela', S)$
(equivalently, there is no homomorphism from $\rela$ to $\rela'$).
It is straightforward to verify that there is an algorithm that,
given an ep-formula, outputs a logically equivalent normalized ep-formula.

{\bf Graphs.}
To every prenex pp-formula $(\rela, S)$
we assign a graph 
whose vertex set is $A \cup S$ 
and where two vertices are connected by an edge if they appear together in a tuple of a relation of $\rela$.
A prenex pp-formula $(\rela, S)$ is called \emph{connected} if its graph is connected. 
A prenex pp-formula $(\rela', S')$ is a
\emph{component} of a prenex pp-formula $(\rela, S)$
over the same signature $\tau$
if there exists a set $C$ that forms a connected component
of the graph of $(\rela, S)$, where:
\begin{itemize}

\item $S' = S \cap C.$

\item For each relation $R \in \tau$, 
a tuple $(a_1, \ldots, a_k)$ is in $R^{\rela'}$
if and only if 
 $(a_1, \ldots, a_k) \in R^{\rela} \cap C^k$.

\end{itemize}
Note that when this holds, the graph of $(\rela', S')$
is the connected component of the graph of $(\rela, S)$ on vertices $C$.
We will use the fact that, if $\phi(V)$ is a prenex pp-formula
and $\phi_1(V_1), \ldots, \phi_c(V_c)$ is a list of its components,
then for any finite structure $\relb$,
it holds that 
$|\phi(\relb)| = \prod_{i=1}^c |\phi_i(\relb)|$.

\begin{example}
\label{ex:pp-formula-components}
Let $\phi$ be the free prenex pp-formula 
from Example~\ref{ex:pp-formula}, and let $(\rela, S)$ be the 
pair representation given there.
The connected components of the graph of $(\rela, S)$
are 
$\{ x, x'\}$, $\{ y, y'\}$, $\{ z \}$,
and
$\{u, v, w\}$.
There are thus four components of the formula $(\rela, S)$;
they are
$(\rela'_{\{x,x'\}},  \{ x, x' \})$,
$(\rela'_{\{y,y'\}},  \{ y \})$, 
$(\rela'_{\{z\}},     \{ z \})$, and
$(\rela'_{\{u,v,w\}}, \emptyset)$
(respectively),
where each $\rela'_C$ is the structure $\rela'$ defined above,
with respect to the set $C$.

Written logically, these four components are
$\psi_1(x, x') =  E(x,x')$,
$\psi_2(y) = \exists y' E(y,y')$,
$\psi_3(z) = \top$, 
and
$\psi_4(\emptyset) = \exists u \exists v \exists w (F(u,v) \wedge G(u,w))$,
respectively.
Here, $\top$ denotes the empty conjunction (considered to be true).
\end{example}



\subsection{Counting complexity}

Throughout, we use $\Sigma$ to denote an alphabet over which
strings are formed.
All problems to be considered are viewed as counting problems.
So, a \emph{problem} is a mapping $Q: \Sigma^* \to \N$.
We view decision problems as problems where,
for each $x \in  \Sigma^*$, it holds that $Q(x)$ is equal to $0$ or $1$.
A \emph{parameterization} is a mapping
$\kappa: \Sigma^* \to \Sigma^*$.
A parameterized problem is a pair $(Q, \kappa)$
consisting of a problem $Q$ and a parameterization~$\kappa$.
Throughout, by $\pi_i$ we denote the operator
that projects a tuple onto its $i$th coordinate.

A partial function $T: \Sigma^* \to \N$
is \emph{polynomial-multiplied} 
with respect to a parameterization~$\kappa$
if there exists a computable function $f: \Sigma^* \to \N$
and a polynomial $p: \N \to \N$ such that,
for each $x \in \dom(T)$, 
it holds that $T(x) \leq f(\kappa(x)) p(|x|)$.

We now give a definition of FPT-computability for partial mappings.

\begin{definition}
Let $\kappa: \Sigma^* \to \Sigma^*$
be a parameterization.
A partial mapping $r: \Sigma^* \to \Sigma^*$
is \emph{FPT-computable} with respect to $\kappa$
if there exist a polynomial-multiplied function $T: \Sigma^* \to \N$
(with respect to $\kappa$)
with $\dom(T) = \dom(r)$
and an algorithm $A$ such that,
for each string $x \in \dom(r)$,
the algorithm $A$ computes $r(x)$ within time $T(x)$;
when this holds, we also say that $r$ is \emph{FPT-computable}
with respect to $\kappa$ \emph{via $A$}.
\end{definition}

As is standard, we may and do freely interchange among
elements of $\Sigma^*$, $\Sigma^* \times \Sigma^*$, and $\N$.
We define FPT to be the class 
that contains a parameterized problem $(Q, \kappa)$
if and only if $Q$ is FPT-computable with respect to $\kappa$.

We now introduce a notion of reduction for counting problems,
which is a form of Turing reduction.
We use $\powfin(A)$ to denote the set containing all finite subsets of
$A$.

\begin{definition}
A \emph{counting FPT-reduction} from a parameterized problem
$(Q, \kappa)$ to another $(Q', \kappa')$
consists of 
 a computable function $h: \Sigma^* \to \powfin(\Sigma^*)$,
and
an algorithm $A$ such that:
\begin{itemize}
\item  on an input $x$,
$A$ may make oracle queries of the form
$Q'(y)$ with $\kappa'(y) \in  h(\kappa(x))$, and
\item $Q$ is FPT-computable with respect to $\kappa$ via $A$.
\end{itemize}
\end{definition}

We use $\clique$ to denote the decision problem  where $(k, G)$ is a
yes-instance when $G$ is a graph that contains a clique of size $k\in
\mathbb{N}$. By $\sclique$ we denote the problem of counting, given
$(k, G)$, the number of $k$-cliques in the graph $G$.
The parameterized versions of these problems,
denoted by $\pclique$ and $\psclique$, are defined
via the parameterization $\pi_1(k,G) = k$.

\subsection{Counting case complexity}

We employ the framework of \emph{case complexity} to
develop some of our complexity results.
We present the needed elements of this
framework for counting problems.
The definitions and results here are due to~\cite{ChenMengel14-pp-arxiv,ChenMengel15-pp-icdt},
are based on the theory of~\cite{Chen14-frontier},
and are presented here for the sake of self-containment;
see those articles for further discussion and motivation of the framework.

The case complexity framework was developed to prove results
on restricted versions of parameterized problems where
not all values of the parameter are permitted.
This type of restricted problem
 arises naturally in query answering problems,
 where one often restricts the queries that are admissible,
 as is done here
(for other examples, see~\cite{DalmauJonsson04-counting,Grohe07-otherside,Chen14-frontier}).

The case complexity framework provides a notion of \emph{case problem}
and a notion of reduction between case problem.  
A case problem was originally~\cite{Chen14-frontier} 
defined as a language $Q$ of pairs (that is, a subset of $\Sigma^* \times \Sigma^*$) where the first element of each pair is ultimately viewed as the parameter, along with a set $S \subseteq \Sigma^*$ restricting the permitted parameter values.
In this article, as we are dealing with counting complexity, in lieu of considering languages, we will consider mappings $\Sigma^* \times \Sigma^* \to \N$.  (Of course, a language of pairs can be naturally viewed as such a mapping by taking its characteristic function.)

One benefit of the framework is that the notion of reduction
does not rely on any form of computability assumption on the sets $S$ involved.  Thus, in comparing case problems using this notion of reduction, 
one does not need to discuss the computability status of these sets $S$, even though in general, it is usual that authors ultimately assume some form of computability on these sets (typically computable enumerability or computability).



Let us turn to the formal presentation of the framework.
A \emph{case problem} consists of a problem
$Q: \Sigma^* \times \Sigma^* \to \N$
and a subset $S \subseteq \Sigma^*$,
and is denoted $Q[S]$.
Note that, although a problem above is defined as a mapping
from $\Sigma^*$ to $\N$, here we work with a problem
that is a mapping from $\Sigma^* \times \Sigma^*$ to $\N$;
this is natural in the current paper, where
an input to the studied problem consists of two parts,
a formula and a structure.  Note that a mapping
$\Sigma^* \times \Sigma^* \to \N$ can be naturally viewed
as a mapping $\Sigma^* \to \N$, as there are natural and well-known
ways to encode
the elements of $\Sigma^* \times \Sigma^*$ as elements of $\Sigma^*$.
For each case problem $Q[S]$, we define
$\param{Q[S]}$ as the parameterized problem $(P, \pi_1)$
where $P(s, x)$ is defined as equal to 
$Q(s, x)$ if $s \in S$, and as $0$ otherwise.





We have the following reduction notion for case problems.

\begin{definition}
A \emph{counting slice reduction} from a case problem $Q[S]$ to a second case problem $Q'[S']$ consists of
\begin{itemize}
 \item a computably enumerable language $U\subseteq \Sigma^* \times \powfin(\Sigma^*)$, and
 \item a partial function 
$r: \Sigma^* \times \powfin(\Sigma^*) \times \Sigma^*\rightarrow \Sigma^*$ 
that has domain $U\times \Sigma^*$ and is 
FPT-computable with respect to $(\pi_1, \pi_2)$ via
an algorithm $A$ that, on input $(s, T, y)$, 
may make queries of the form $Q'(t,z)$
where $t \in T$,
\end{itemize}
such that the following conditions hold:
\begin{itemize}
 \item (coverage) for each $s\in S$, there exists $T \subseteq S'$
such that $(s, T)\in U$, and 
 \item (correctness) for each $(s, T) \in U$, it holds (for each $y\in \Sigma^*$) that 
$Q(s,y) = r(s, T, y).$
\end{itemize}
\end{definition}

Let us provide some intuition for this definition.
Here, when discussing an instance $(s, y)$ of a case problem,
we refer to the first part $s$ as the parameter.
The role of $U$ is to provide all pairs $(s, T)$ such that
instances (of the first problem) with parameter $s$ 
can be reduced to instances (of the second problem) 
whose parameters lie in $T$.
Correspondingly, the \emph{coverage} condition posits that
each $s \in S$ is covered by the second set $S'$ in the sense that
there exists a pair $(s, T) \in U$ with $T \subseteq S'$.
The partial function $r$ is the actual reduction; given
a pair $(s, T) \in U$ along with a string $y$, it computes
the value $Q(s,y)$---this is what the \emph{correctness} condition asserts.
As here in this article we are dealing with counting complexity, we permit
a form of Turing reduction;
so, the algorithm $A$ of the partial function $r$, upon being given
a triple $(s, T, y)$,
may make (possibly multiple) 
queries to the second problem, so long as the queries
are about instances whose parameter falls into $T$.

We have the following key property of counting slice reducibility.

\begin{theorem} \label{thm:transitivity}
\cite{ChenMengel14-pp-arxiv}
Counting slice reducibility is transitive.
\end{theorem}

The following theorem
shows that, from a counting slice reduction,
one can obtain 
complexity results for the corresponding
parameterized problems.

\begin{theorem}
\label{thm:slice-red-gives-fpt-red}
\cite{ChenMengel14-pp-arxiv}
Let $Q[S]$ and $Q'[S']$ be case problems.
Suppose that $Q[S]$ counting slice reduces to $Q'[S']$,
and that both $S$ and $S'$ are computable.
Then $\param{Q[S]}$ counting FPT-reduces to $\param{Q'[S']}$.
\end{theorem}

\subsection{Classification of pp-formulas}
\label{sct:ppclassification}

We present the complexity classification of pp-formulas
previously presented in~\cite{ChenMengel14-pp-arxiv,ChenMengel15-pp-icdt}.
The following definitions are adapted from that article.
Let $(\rela,S)$ be a prenex pp-formula,
let $\reld$ be the core thereof, and let $G = (D,E)$ be the
graph of $\reld$.
An
\emph{$\exists$-component} of $(\rela,S)$ is a graph of the form 
$G[V']$ where there exists
$V \subseteq D$ that is the vertex set of a component of
$G[D \setminus S]$ and $V'$ is the union of $V$
with all vertices in $S$ having an edge to $V$.
Define $\contract(\rela,S)$ 
to be the graph on vertex set $S$
obtained by starting from $G[S]$ 
and adding an edge between any two vertices that appear together
in an $\exists$-component of $(\rela,S)$.

Let $\Phi$ be a set of prenex pp-formulas.
Let us say that $\Phi$ satisfies the
\emph{contraction condition} if
the graphs in the set $\contract(\Phi) := \{ \contract(\phi) ~|~ \phi \in \Phi \}$
are of bounded treewidth.
Let us say that $\Phi$ satisfies the
\emph{tractability condition}
if it satisfies the contraction condition 
and, in addition,
the cores of $\Phi$ are of bounded treewidth;
here, the treewidth of a prenex pp-formula is defined
as that of its graph. 
We omit the definition of treewidth,
as it is both well-known and not needed
to understand the main technical proof of this 
article (which is in Section~\ref{sect:proof-equiv-theorem}).


\begin{definition}
We define $\countp$ to be the problem that maps a pair
$(\phi(V),\relb)$ consisting of a fo-formula and a finite structure
to the value $|\phi(\relb)|$.
\end{definition}

\begin{theorem}
\label{thm:pp-tractability-condition} \cite{ChenMengel14-pp-arxiv}
Let $\Phi$ be a set of prenex pp-formulas
that satisfies the tractability condition.
Then, the restriction of $\param{\countp[\Phi]}$ to $\Phi \times \Sigma^*$
is an FPT-computable partial function.
\end{theorem}

\begin{theorem} \cite{ChenMengel14-pp-arxiv}
\label{thm:pp-complexity-results}
Let $\Phi$ be a set of prenex pp-formulas of bounded arity
that does not satisfy the tractability condition.
\begin{enumerate}


\item If $\Phi$ satisfies the contraction condition,
then it holds that
$\countp[\Phi]$ and $\clique[\mathbb{N}]$ are interreducible,
  under
counting slice reductions.

\item 
Otherwise,
 there exists a counting slice reduction 
 from $\sclique[\mathbb{N}]$ to $\countp[\Phi]$.

\end{enumerate}
\end{theorem}

We say that a set of formulas $\Phi$ has \emph{bounded arity} if
there exists a constant $k \geq 1$ that upper bounds the
arity of each relation symbol appearing in a formula in $\Phi$.

\section{Main theorems}


The following theorem, which we call the \emph{equivalence theorem}
and which is proved in Section~\ref{sect:proof-equiv-theorem},
is our  primary technical result; it is used to derive our
complexity trichotomy on ep-formulas from the
known complexity trichotomy on pp-formulas
(which was presented in Section~\ref{sct:ppclassification}).

\begin{theorem}
\label{thm:equivalence-theorem} 
(Equivalence theorem)
Let $\Phi$ be a set of ep-formulas.  
There exists a set $\Phi^+$ of prenex pp-formulas
with the following property: the two counting case problems
$\countp[\Phi]$ and $\countp[\Phi^+]$
are interreducible under counting slice reductions.
In particular, there exists an algorithm that computes,
given an ep-formula $\phi$, a finite set $\phi^+$ of prenex pp-formulas
such that for any set $\Phi$ of ep-formulas, the 
set $\Phi^+$ defined as $\bigcup \{ \phi^+ ~|~ \phi \in \Phi \}$
has the presented property.
\end{theorem}

We now state our trichotomy theorem
on the complexity of counting answers to ep-formulas,
and show how to prove it using
the equivalence theorem.

\begin{theorem}
\label{thm:trichotomy}
(Trichotomy theorem)
Let $\Phi$ be a computable set of ep-formulas of bounded arity,
and let $\Phi^+$ be the set of pp-formulas
given by Theorem~\ref{thm:equivalence-theorem}.
\begin{enumerate}

\item If $\Phi^+$ satisfies the tractability condition,
then it holds that $\param{\countp[\Phi]}$ is in FPT.

\item If $\Phi^+$ does not satisfy the tractability condition
but satisfies the contraction condition,
then it holds that $\param{\countp[\Phi]}$ is interreducible with
$\pclique$ under counting FPT-reduction.

\item Otherwise, there is a counting FPT-reduction from
the problem 
$\psclique$ to $\param{\countp[\Phi]}$.

\end{enumerate}
\end{theorem}

\begin{proof}
For (1), 
we use the counting slice reduction $(U, r)$ from
$\countp[\Phi]$ to $\countp[\Phi^+]$
given by Theorem~\ref{thm:equivalence-theorem}.
In particular, given as input $(\phi, \relb)$,
it is first checked if $\phi \in \Phi$; if not, $0$ is output.
Otherwise,
the algorithm for $r$
is invoked on $(\phi, \phi^+, \relb)$,
where $\phi^+$ is as defined in 
the statement of Theorem~\ref{thm:equivalence-theorem};
queries to $\countp(\psi, \relb)$ where $\psi \in \Phi^+$ are
resolved according to the algorithm of
Theorem~\ref{thm:pp-tractability-condition}.

For (2) and (3), we make use of the result
(Theorem~\ref{thm:equivalence-theorem})
that the problems $\countp[\Phi]$ and $\countp[\Phi^+]$ 
are interreducible under counting slice reductions.
For (2), we have from Theorem~\ref{thm:pp-complexity-results}
that $\countp[\Phi^+]$ and $\clique[\N]$ are interreducible under
counting slice reductions.
Hence, we obtain that the problems $\clique[\N]$ and $\countp[\Phi]$ are interreducible
under counting slice reductions, and the result follows from 
Theorem~\ref{thm:slice-red-gives-fpt-red}.
For (3), we have from Theorem~\ref{thm:pp-complexity-results}
that there is a counting slice reduction
from $\sclique[\N]$ to $\countp[\Phi^+]$, and hence
from $\sclique[\N]$ to $\countp[\Phi]$;
the result then follows from 
Theorem~\ref{thm:slice-red-gives-fpt-red}.
\end{proof}

Let us remark that when case (2) applies, 
a consequence of this theorem is that
the problem 
$\param{\countp[\Phi]}$ is not in FPT unless W[1] is in FPT,
since $\pclique$ is W[1]-complete;
in a similar fashion, when case (3) applies,
the problem 
$\param{\countp[\Phi]}$ is not in FPT unless $\sh$W[1] is in FPT,
since $\psclique$ is $\sh$W[1]-complete.


\section{Examples}\label{sct:examples}

Before proving the equivalence theorem in full generality,
we discuss some example
ep-formulas to illustrate and preview 
some of the issues and difficulties
with which the argument needs to cope.

\begin{example}\label{ex:eqfirst}
 Consider the formula 
 \[\phi(w,x,y,z) := E(x,y) \land (E(w,x) \lor (E(y, z) \land E(z,z))).\]
 As a first simplification step, we bring 
 disjunction to the outermost level in $\phi$:
 \begin{align*} &\phi(w,x,y,z) \\&\equiv(E(x,y) \land E(w,x)) \lor (E(x,y) \land E(y,z) \land E(z,z)).\end{align*}
 Now let us set $\phi_1(w,x,y,z)\equiv E(x,y) \land E(w,x)$ and also set $\phi_2(w,x,y,z)\equiv E(x,y) \land E(y,z) \land E(z,z)$.
 We can use inclusion-exclusion to count the number of satisfying assignments of $\phi$ on a structure $\relb$ by 
 \[|\phi(\relb)| = |\phi_1(\relb)| + |\phi_2(\relb)| - |(\phi_1\land \phi_2)(\relb)|.\]

One point to observe is that, in this last expression,
the count $|\phi_1(\relb)|$ needs to be determined with respect to
its set of liberal variabes $\lib(\phi_1) = \{ w, x, y, z \}$, even though $z$
does not appear in any atom of $\phi_1$.
If the count $|\phi_1(\relb)|$ is not computed in this way,
the above expression for $|\phi(\relb)|$ fails
to hold in general.
The situation is analogous for the formula $\phi_2$, where 
$w$ does not appear in any atom.
 %
\end{example}

\begin{example}\label{ex:eqsecond}
 In general, if we are given an ep-for\-mula $\phi= \phi_1 \lor \ldots \lor \phi_n$ where the $\phi_i$ are pp-formulas, then to compute the count 
 $|\phi(\relb)|$ of $\phi$ relative to $\relb$,
 it suffices to know the count  for each of the $2^n - 1$ pp-formulas 
 obtained by taking a conjunction of a non-empty subset of the $\phi_i$. 
 In this example, we will see that, in fact, one does not always need to consider all of these conjunctions. 
 To this end, set $V = \{ w,x,y,z \}$ and set
$$\phi(V)= \phi_1(V) \lor \phi_2(V) \lor \phi_3(V)$$ 
where 
$\phi_1(V) = E(x,y) \land E(y,z)$, 
$\phi_2(V) = E(z,w) \land E(w,x)$ and 
$\phi_3(V) = E(w,x) \land E(x,y)$. 
Applying in\-clu\-sion-ex\-clu\-sion, we obtain
 \begin{align*}|\phi(\relb)| = &|\phi_1(\relb)| + |\phi_2(\relb)| + |\phi_3(\relb)|\\
  &- |(\phi_1\land \phi_2)(\relb)| - |(\phi_1\land \phi_3)(\relb)|  \\&- |(\phi_2 \land \phi_3)(\relb)|+ |(\phi_1 \land \phi_2 \land \phi_3)(\relb)|. \end{align*}

 Now observe that the formulas $\phi_1$, $\phi_2$ and $\phi_3$ 
are actually equivalent to each other up to renaming variables;
consequently,  these formulas
are equivalent in that,
for any structure $\relb$,
they
yield the same count: $|\phi_1(\relb)| = |\phi_2(\relb)| = |\phi_3(\relb)|$.
In Section~\ref{subsect:counting-equivalence}, we formalize
and give a characterization of this notion of equivalence
(on pp-formulas).
The formulas $\phi_1 \land \phi_3$ and $\phi_2 \land \phi_3$
are also equivalent in this sense.  We may thus obtain
the following expression for $|\phi(\relb)|$.
 \begin{align*}|\phi(\relb)| = &3 \cdot |\phi_1(\relb)| - |(\phi_1\land \phi_2)(\relb)| \\&- 2\cdot |(\phi_1\land \phi_3)(\relb)| + 
 |(\phi_1 \land \phi_2 \land \phi_3)(\relb)|. \end{align*}

So far, we have only unified formulas that are equivalent up to renaming variables.  
In our parameterized complexity setting where $\phi$ is the parameter,
this does not yield a significant decrease
in the complexity of computing  $|\phi(\relb)|$.
However, we will now observe a simplification that is more substantial
in this sense.  
Namely, one can verify that the formulas
$\phi_1 \land \phi_2$ and $\phi_1\land \phi_2\land \phi_3$
are identical.  So, if we identify their terms in this last 
expression for $|\phi(\relb)|$, we obtain a cancellation
and arrive to the following expression:
\[|\phi(\relb)| = 3 \cdot |\phi_1(\relb)| - 2\cdot |(\phi_1\land \phi_3)(\relb)| . \]

The savings obtained by observing this cancellation are 
significant, in the following sense.
The  pp-formulas $\phi_1 \land \phi_2$
and $\phi_1 \land \phi_2 \land \phi_3$, 
which were cancelled,
were the only formulas in the expression for $|\phi(\relb)|$
which did not have treewidth $1$; they had treewidth $2$.
As it is known that the runtime of evaluation algorithms
for quantifier-free pp-formulas scales with their treewidth~\cite{Marx10-canyoubeat},
this reduction in treewidth yields a superior runtime for evaluating 
$|\phi(\relb)|$.
%
%
\end{example}

As we have seen in the above examples, 
counting on an ep-formula can, via inclusion-exclusion,
reduce to counting on a finite set of pp-formulas.
(This is carried out in our argument; see Section~\ref{subsect:all-free}).
As just seen in Example~\ref{ex:eqsecond}, there can be some subtlety
in choosing a desirable set of pp-formulas to reduce to.
One question not addressed so far is how one can reduce
from counting on a such obtained set of pp-formulas
to counting on the original ep-formula.
To this end, let us revisit our first example.

\begin{example}\label{ex:eqfirstrevisited}
 Let us consider again the formulas of Example~\ref{ex:eqfirst}. 
  Assume that we are given access to an oracle 
  that lets us compute $|\phi(\reld)|$,
  for any structure $\reld$ of our choice. 
  We will see that, given a structure $\relb$, 
  we can compute $|\phi_1(\relb)|$, $|\phi_2(\relb)|$, and $|(\phi_1\land \phi_2)(\relb)|$ efficiently using this oracle. 

  To see this, consider the structure $\relc$ with 
universe $C = \{ 1, 2, 3, 4 \}$ and
  $E^\relc = \{(1,2),(2,3),(3,4),(4,4)\}$. It is easy to check that the formulas $\phi_1$, $\phi_2$ and $\phi_1\land \phi_2$ all have a different number of answers with respect to $\relc$.  
  Now note 
  that for every pp-formula $\psi$ and every pair of structures 
  $\reld_1$, $\reld_2$ we have 
  $|\psi(\reld_1\times \reld_2)|= |\psi(\reld_1)| \cdot |\psi(\reld_2)|$. Querying the oracle for $|\phi(\cdot)|$ 
  on $\relb \times \relc^i$ for 
  the values $i=0, 1, 2$,
   we obtain the linear system
 \[ A\begin{pmatrix} |\phi_1(\relb)|\\ |\phi_2(\relb)|\\ -|(\phi_1\land \phi_2)(\relb)|\end{pmatrix} = \begin{pmatrix}(\phi(\relb)\\ \phi(\relb \times\relc) \\ \phi(\relb \times\relc^2)\end{pmatrix}\] with \[A = \begin{pmatrix} 1 & 1 & 1 \\ |\phi_1(\relc)| & |\phi_2(\relc)| & |(\phi_1\land \phi_2)(\relc)|\\ |\phi_1(\relc)|^2 & |\phi_2(\relc)|^2 & |(\phi_1\land \phi_2)(\relc)|^2
    \end{pmatrix}.\]
    Note that the entries of 
    $A$
    can be computed efficiently, and  the vector on the right-hand-side of the equation
    can be provided by our oracle. The matrix $A$ is a Vandermonde matrix,
    as a consequence of the choice of $\relc$. 
Thus, the system has a unique solution and can be solved to 
 determine $|\phi_1(\relb)|$, $|\phi_2(\relb)|$, and $|(\phi_1\land \phi_2)(\relb)|$, as desired.
\end{example}

In Example~\ref{ex:eqfirstrevisited} we have seen that,
for the particular ep-formula $\phi$ discussed, counting on $\phi$ is 
in a certain sense interreducible with counting
on the pp-formulas 
$$\{ \phi_1, \phi_2, \phi_1 \wedge \phi_2 \}.$$
The statement of the
equivalence theorem (Theorem~\ref{thm:equivalence-theorem})
asserts that for \emph{any} ep-formula $\phi$, there 
exists a finite set $\phi^+$ of pp-formulas such that
one has this interreducibility.

\section{Proof of equivalence theorem}
\label{sect:proof-equiv-theorem}


In this section, we give a decidable characterization
of \emph{counting equivalence} 
(Section~\ref{subsect:counting-equivalence});
we then study a relaxation thereof which we call
\emph{semi-counting equivalence}
(Section~\ref{subsect:semi-counting-equivalence});
we prove the equivalence theorem
in the particular case of all-free ep-formulas
(Section~\ref{subsect:all-free});
and, we end by proving the equivalence theorem
in its full generality
(Section~\ref{subsect:general-equivalence-theorem}).
Throughout this section, we generally assume pp-formulas to be prenex.

\subsection{Counting equivalence}
\label{subsect:counting-equivalence}

As we have seen in the examples of Section~\ref{sct:examples}, it will be important to see when two different pp-formulas give same number of answers for every structure, because it will allow us to make simplifications in formulas we get by inclusion-exclusion. To this end, we make the following definition.

\begin{definition}
Define two fo-formulas $\phi(V), \phi'(V')$ to be
 \emph{counting equivalent} if they are over the same vocabulary $\tau$
 and for each finite $\tau$-structure $\relb$ 
 it holds that $|\phi(\relb)| = |\phi'(\relb)|$.
\end{definition}

In this subsection, we characterize counting equivalence 
for pp-formulas. To approach the characterization, we start off with an example.

\begin{example}\label{ex:rename}
It is apparent that logically equivalent formulas are counting equivalent, but the converse direction is not true. To see this, consider the pp-formulas $\phi_1(x,y) = E(x, y)$ and $\phi_2(w,z)= E(w,z)$. Obviously, $\phi_1$ and $\phi_2$ are counting equivalent (they just count the number of tuples in the relation $E$ of a structure $\relb$). But $\phi_1$ and $\phi_2$ are not logically equivalent;
indeed, the assignments in $\phi_1(\relb)$ and $\phi_2(\relb)$ assign values to different variables. 

Note that one way of witnessing the counting equivalence of $\phi_1$ and $\phi_2$ is simply renaming the variable $w$ to $x$ and $z$ to $y$ to get equivalent formulas. Since this syntactic renaming 
obviously does not change the number of satisfying assignments, 
one can conclude that $\phi_1$ and $\phi_2$ are counting equivalent.
\end{example}

Example~\ref{ex:rename} motivates the following definition.

\begin{definition}
We say that
two pp-formulas
$$(\rela, S), (\rela', S')$$
 over the same signature
 are \emph{renaming equivalent}
if there exist surjections 
$h: S \to S'$ and $h': S' \to S$
 that can be extended to homomorphisms 
 $\bar{h} : \rela \rightarrow \rela'$ and 
 $\bar{h'}:\rela'\rightarrow \rela$, respectively.
\end{definition}

Informally speaking, on pp-formulas, 
%
two formulas are renaming equivalent
if they become logically equivalent after a renaming of variables,
as occurred in 
Example~\ref{ex:rename}.
Hence, renaming equivalence is a relaxation of logical equivalence.
Recall that logical equivalence of pp-formulas was characterized,
in Theorem~\ref{thm:ChandraMerlin}.

The main theorem of this subsection is that renaming equivalence does not only imply counting equivalence but is actually equivalent to it.

\begin{theorem}\label{thm:renamingandcounting}
Two pp-formulas 
$$\phi_1(S_1), \phi_2(S_2)$$ are counting equivalent
if and only if they are renaming equivalent.
\end{theorem}

Note that Theorem~\ref{thm:renamingandcounting} gives a syntactic/algebraic characterization of counting equivalence which makes counting equivalence decidable by a straightforward algorithm and in fact even puts it into $\mathrm{NP}$.

Before we prove Theorem \ref{thm:renamingandcounting}, 
 we start off with an simple observation that will be helpful in the proof.

\begin{observation}\label{obs:samevariables}
 Let $\phi$ and $\phi'$ be counting equivalent pp-formulas. Then $|\lib(\phi)| =|\lib(\phi')|$.
\end{observation}

\begin{proof}
 Let $\relc$ be a structure that interprets every relation symbol in $R$ of $\phi$ by $R^\relc := \{0,1\}^{\mathsf{arity}(R)}$. Then $|\phi(\relc)| = 2^{|\lib(\phi)|}$ and $|\phi'(\relc)|= 2^{|\lib(\phi')|}$ and the claim follows directly.
\end{proof}

\begin{proof} (Theorem~\ref{thm:renamingandcounting})
We begin with the backward direction;
let $h_1: S_1 \to S_2$ and $h_2: S_2 \to S_1$ be the surjections
from the definition of renaming equivalence.
The existence of these surjections implies that $|S_1| = |S_2|$
and that each of $h_1$, $h_2$ is a bijection.
Let $\relb$ be an arbitrary structure.
For each $f: S_2 \to B$ in  $\phi_2(\relb)$,
it is straightforward to verify that the
composition $f(h_1)$ is in $\phi_1(\relb)$.
Since the mapping that takes each such $f$ to $f(h_1)$
is injective (due to $h_1$ being a bijection), 
we obtain that $|\phi_1(\relb)| \geq |\phi_2(\relb)|$.
By symmetric reasoning, we can obtain that
$|\phi_1(\relb)| \leq |\phi_2(\relb)|$,
and we conclude that
$|\phi_1(\relb)| = |\phi_2(\relb)|$.

 For the other direction, let $\phi_1(S_1)$ and $\phi_2(S_2)$ be two pp-formulas over a common vocabulary $\tau$ 
 that are not renaming equivalent; 
 let $(\rela_1, S_1)$ and $(\rela_2, S_2)$ 
 be the corresponding structures. 
 By way of contradiction, assume that $\phi_1$ and $\phi_2$ are counting equivalent. 
 If it holds that
 $|\lib(\phi_1)| \ne |\lib(\phi_2)|$, we are done by 
 Observation~\ref{obs:samevariables}. So we may assume, after potentially renaming some variables, that $\lib(\phi_1)=\lib(\phi_2)=:S$.
 
When $\relc$, $\reld$ are structures 
with $S \subseteq C \cap D$,
let us define $\hom(\relc, \reld, S)$ to be the set of mappings
from $S$ to $D$ that can be extended to a homomorphism 
from $\relc$ to $\reld$;
 denote by $\surj(\relc, \reld, S)$ 
 the surjections $h :S\rightarrow S$ that lie in $\hom(\relc, \reld, S)$.

 As $(\rela_1, S_1)$ and $(\rela_2, S_2)$ are by hypothesis
 not renaming equivalent, 
 we may assume, without loss of generality, that $\surj(\rela_1, \rela_2, S)= \emptyset$. For $T\subseteq S$ let us use
 $\hom_T(\rela_1, \rela_2, S)$ to denote the set of mappings $h\in \hom(\rela_1, \rela_2, S)$ such that $h(S)\subseteq  T$. By inclusion-exclusion we get 
 \[|\surj({\phi_1},{\phi_2},S)| = \sum_{T\subseteq S} (-1)^{|S|-|T|} |\hom_T(\rela_1, \rela_2,S)|.\]
 
For $i\ge 0$ let $\hom_{i,T}(\rela_1, \rela_2,S)$ be the set of mappings $h\in \hom(\rela_1, \rela_2,S)$ such that $h$ maps exactly $i$ variables from $S$ into $T$. Now for each $j=1,\ldots , |S|$ we construct a new structure $\reld_{j,T}$ over the domain $D_{j,T}$. To this end, let $a^{(1)}, \ldots, a^{(j)}$ be copies of $a\in T$ that are not in $A_2$. Then we set \begin{align*}D_{j,T}:= \{a^{(k)}\mid a\in A_2, a\in T, k\in [j]\} \cup (A_2\setminus T).\end{align*} We define a mapping $B:A_2\rightarrow \calP(D_{j,T})$, where $\calP(D_{j,T})$ is the power set of $D_{j,T}$, by
\[B(a):= \begin{cases}
            \{a^{(k)}\mid k\in [j]\}\}, & \text{if } a\in T\\
            \{a\}, & \text{otherwise}.
           \end{cases}\]
For every relation symbol $R\in \tau$ we define
\[ R^{\reld_{T,j}} := \bigcup_{(d_1, \ldots, d_s)\in R^{\rela_2}} B(d_1)\times \ldots \times B(d_s).\]

Then every $h\in \hom_{i,T}(\rela_1, \rela_2,S)$ corresponds to $j^i$ mappings in $\hom(\rela_1, \reld_{j,T},S)$. Thus for each $j$ we get \[\sum_{i=1}^{|S|} j^i |\hom_{i,T}(\rela_1, \rela_2,S)| = |\hom(\rela_1,\reld_{j,T}, S)|.\] This is a linear system of equations and the corresponding matrix is a Vandermonde matrix; consequently, the value $\hom_{T}(\rela_1, \rela_2,S) = \hom_{|S|,T}(\rela_1, \rela_2,S)$ can \arxivversion{efficiently} be computed from $|\hom(\rela_1, \reld,S)|=|\phi_1(\reld)|$ for some structures $\reld$. 
We can similarly determine 
$|\hom_T(\rela_2, \reld, S)|$
as a function of $|\phi_2(\reld)|$ for the same structures $\reld$.
Since $|\phi_1(\reld)| = |\phi_2(\reld)|$ for every structure $\reld$ by assumption, it follows that for every subset $T\subseteq S$ we have 
\[|\hom_{T}(\rela_1, \rela_2,S)| = |\hom_{T}(\rela_2, \rela_2,S)|.\] 
But then we have 
\[|\surj(\rela_1, \rela_2, S)| = |\surj(\rela_2, \rela_2, S)|.\] 
Since 
$\surj(\rela_1, \rela_2, S)= \emptyset$ and $\id\in \surj(\rela_2, \rela_2, S)$, this is a contradiction. Consequently, we obtain that $\phi_1$ and $\phi_2$ are not counting equivalent.
\end{proof}

\subsection{Semi-counting equivalence}
\label{subsect:semi-counting-equivalence}

In this subsection, we study a
relaxation of the notion of \emph{counting equivalence}. This notion will be necessary when we emulate the approach of Example~\ref{ex:eqfirstrevisited} in the proof of the Equivalence theorem: we will again construct a system of linear equations that we want to solve. In order to ensure solvability, we will make sure that the matrix of the system is again a Vandermonde matrix which in particular means that all its entries must be positive. Consequently, since the entries are of the form $|\phi(\relc)|^k$ for pp-formulas $\phi$ some carefully chosen structure $\relc$ and integers $k$, it will be necessary to understand counting equivalence in the case where $\phi(\relc)$ is non-empty. The necessary notion is formalized by the following definition.

\begin{definition}
Call two prenex pp-formulas 
$$\phi_1(V_1), \phi_2(V_2)$$
on the same vocabulary
\emph{semi-counting equivalent}
if for each finite structure $\relb$
such that $|\phi_1(\relb)| > 0$ and $|\phi_2(\relb)| > 0$,
it holds that $|\phi_1(\relb)| = |\phi_2(\relb)|$.
\end{definition}

\begin{example}
 The pp-formulas $\phi_1(x,y)= E(x,y)$ and 
 $\phi_2(x,y)=\exists z (E(x,y)\land F(z))$ are not counting equivalent, because for every structure $\relb$ for which $F^\relb=\emptyset$, we have $|\phi_2(\relb)| = 0$ while $|\phi_1(\relb)|$ may be non-zero if $E^{\relb}$ is non-empty. But if we have for a structure $\relb$ such that $|\phi_2(\relb)|>0$, then $F^\relb\ne \emptyset$ and it is straightforward to verify that $|\phi_1(\relb)| = |\phi_2(\relb)|$. Consequently,
 we have that $\phi_1$ and $\phi_2$ are semi-counting equivalent.
\end{example}

For each free prenex pp-formula $\phi(V)$, define $\widehat{\phi}(V)$ 
to be the pp-formula
obtained from $\phi$ 
by removing each atom
that occurs in a non-liberal component of $\phi$
(a component of $\phi$ not having liberal variables).

\begin{example}
Consider the pp-formula $\phi$ discussed
 in Examples~\ref{ex:pp-formula}
and~\ref{ex:pp-formula-components}.
This pp-formula has $4$ components, namely, the pp-formulas
$\psi_1(x, x')$,
$\psi_2(y)$,
$\psi_3(z)$,
and
$\psi_4(\emptyset)$
defined in Example~\ref{ex:pp-formula-components}.
The formulas $\psi_1$, $\psi_2$, and $\psi_3$ are liberal,
but the formula $\psi_4$ is not liberal.
Recall that the formula
$\phi(x,x',y,z)$ is equal to 
$$\exists y' \exists u \exists v \exists w (E(x, x') \wedge E(y,y') \wedge F(u,v) \wedge G(u,w))$$
and that we have
$\psi_4(\emptyset) = \exists u \exists v \exists w (F(u,v) \wedge G(u,w))$.
We hence have that $\widehat{\phi}(x,x',y,z)$ is the formula
$$\exists y' \exists u \exists v \exists w 
(E(x, x') \wedge E(y,y')).$$
\end{example}

The following characterization of semi-counting equivalence
 is the main theorem of this subsection.

\begin{theorem}\label{thm:semiequichar}
Let $\phi_1(V_1), \phi_2(V_2)$ be two free prenex pp-formulas.
It holds that $\phi_1(V_1)$ and $\phi_2(V_2)$
are semi-counting equivalent
if and only if $\widehat{\phi_1}(V_1)$ and $\widehat{\phi_2}(V_2)$ are counting equivalent.
\end{theorem}

We will use the following proposition in the proof of Theorem~\ref{thm:semiequichar}.

\begin{proposition}\label{prop:samesolutions}
 Let $\phi(V)$ be a free prenex pp-for\-mula. Then for every structure $\relb$ we have $\phi(\relb)=\emptyset$ or $\phi(\relb)=\widehat{\phi}(\relb)$.
\end{proposition}

\begin{proof}
 Let $\relb$ be a structure. 
Let $\psi$ be the conjunction of the components deleted from $\phi$
to obtain $\widehat{\phi}$.
If $\psi$ is false on $\relb$, then obviously $\phi(\relb) = \emptyset$.
Otherwise, $\psi$ is true on $\relb$, and for any assignment $f: V \to B$,
it holds that $\relb, f \models \phi$ if and only if $\relb, f \models \widehat{\phi}$.
\end{proof}

\begin{proof}  (Theorem~\ref{thm:semiequichar})
 Assume first that $\widehat{\phi_1}$ and $\widehat{\phi_2}$ are counting equivalent. Let $\relb$ be a structure. Then if $|\phi_1(\relb)|> 0$ and $|\phi_2(\relb)|> 0$, we have by Proposition~\ref{prop:samesolutions} and counting equivalence of $\widehat{\phi_1}$ and $\widehat{\phi_2}$ that 
$|\phi_1(\relb)| = |\widehat{\phi_1}(\relb)| = |\widehat{\phi_2}(\relb)| = |\phi_2(\relb)|$,
 so $\phi_1$ and $\phi_2$ are semi-counting equivalent.
 
 For the other direction let now $\phi_1$ and $\phi_2$ be semi-counting equivalent. By way of contradiction, we assume that $\widehat{\phi_1}$ and $\widehat{\phi_2}$ are not counting equivalent. Then by definition there is a structure $\relb$ such that $|\widehat{\phi_1}(\relb)|\ne |\widehat{\phi_2}(\relb)|$. 
 Note that each component of $\widehat{\phi_1}$ and $\widehat{\phi_2}$
  has a liberal variable.
 
Let $\reli = \reli_{\tau}$,
where $\tau$ is the vocabulary of $\phi_1$ and $\phi_2$.
For each $k\in \mathbb{N}$ we denote by $\relb+k\reli$ the structure we get from $\relb$ by disjoint union with $k$ copies of $\reli$. Note that for $k>0$ we have $|\phi(\relb +k\reli)| > 0$ for every pp-formula $\phi$. Consequently, for every $k>0$ we have $|\phi_1(\relb+k\reli)| = |\widehat{\phi_1}(\relb+k\reli)|$ and $|\phi_2(\relb+k\reli)| = |\widehat{\phi_2}(\relb+k\reli)|$ by Proposition~\ref{prop:samesolutions}. By the semi-counting equivalence of $\phi_1$ and $\phi_2$ we also have $|\phi_1(\relb +k\reli)| =  |\phi_2(\relb +k\reli)|$ for all $k>0$. It follows that $|\widehat{\phi_1}(\relb +k\reli)| =  |\widehat{\phi_2}(\relb +k\reli)|$ for $k>0$.

Let $\phi_{1,1}, \ldots, \phi_{1,n}$ denote 
the components of $\widehat{\phi_1}$, and 
let $\phi_{2,1}, \ldots , \phi_{2,m}$ denote 
the  components of $\widehat{\phi_2}$. 
Because every component of $\widehat{\phi_1}$ has a liberal variable, 
we have 
\begin{align*}
|\widehat{\phi_1}(\relb+k\reli)| &= \sum_{J\subseteq [n]} k^{n-|J|} \prod_{j\in J}|\phi_{1,j}(\relb)| \\
 & =\sum_{\ell=0}^n k^{n-\ell} \sum_{J\subseteq [n], |J|=\ell}  \prod_{j\in J}|\phi_{1,j}(\relb)|.
 \end{align*}
We can express
$|\widehat{\phi_2}(\relb+k\reli)|$ analogously.
The expressions are polynomials in $k$ and they are equal for every positive integer $k$ by the observations above;
thus the coefficients of the polynomials must coincide.
The coefficients of $k^0$, namely 
the values
$\prod_{j\in [n]}|\phi_{1,j}(\relb)|$ and
$\prod_{j\in [m]}|\phi_{2,j}(\relb)|$, are thus equal. But then we get 
\[|\widehat{\phi_1}(\relb)|= \prod_{j\in [n]}|\phi_{1,j}(\relb)|= \prod_{j\in [m]}|\phi_{2,j}(\relb)|=|\widehat{\phi_2}(\relb)|,\] which is a contradiction to our assumption.
\end{proof}

\begin{corollary}
 Semi-counting equivalence 
 is an equivalence relation
  (on pp-formulas).
\end{corollary}

We now present a lemma that will be of utility; it is proved by induction.


\begin{lemma}\label{lem:semigeneral}
 Let $\phi_1(S_1), \ldots, \phi_n(S_n)$ be pp-formulas over the same vocabulary $\tau$, 
 which are liberal (that is, with each $|S_i| > 0$). Then there is a structure $\relc$ (over $\tau$) such that 
\begin{itemize}
\item for all pp-formulas $\phi$ (over $\tau$) it holds
that $|\phi(\relc)|>0$, and 

\item
for all $i,j\in [n]$ such that $\phi_i$ and $\phi_j$ are not semi-counting equivalent, it holds that $|\phi_i(\relc)|\ne |\phi_j(\relc)|$.
\end{itemize}
\end{lemma}

In order to establish this lemma, we first prove the following lemma.

\begin{lemma}\label{lem:twocasesemi}
 Let $\phi_1(S_1)$ and $\phi_2(S_2)$ be two pp-for\-mu\-las 
 over a vocabulary $\tau$ that are not semi-counting equivalent. Then there is a structure $\reld$ such that for every primitive positive formula $\phi$ over $\tau$ we have $|\phi(\reld)|> 0$ and $|\phi_1(\reld)|\ne |\phi_2(\reld)|$.
\end{lemma}

\begin{proof}
 Let $\relb$ be any structure on which $\phi_1$ and $\phi_2$ have a non-zero but different number of solutions. Such a structure exists by definition of semi-counting equivalence. We claim that we can choose $\reld = \relb + k\reli$ for some $k\in \nats, k>0$ where $\relb +k\reli$ is defined as in the proof of Theorem~\ref{thm:semiequichar}. By way of contradiction, assume that $|\phi_1(\relb +k\reli)| =  |\phi_2(\relb +k\reli)|$ for all $k\in \mathbb{N}, k>0$. Then with the same argument as in the proof of Theorem~\ref{thm:semiequichar} we get the contradiction that $|\phi_1(\relb)| =  |\phi_2(\relb)|$. 
\end{proof}

\begin{proof} (Lemma~\ref{lem:semigeneral})
We prove this by induction on $n$;
the case $n=2$ is implied by Lemma~\ref{lem:twocasesemi}.

When $n > 2$, we first observe that it suffices to prove the result 
when the $\phi_i$ are pairwise not semi-counting equivalent,
so we assume that this holds.
Let $\reld$ be the structure that we get by induction 
for
$\phi_1, \ldots, \phi_{n-1}$.
We may assume w.l.o.g.~that 
$$|\phi_1(\reld)|<|\phi_2(\reld)|<\ldots < |\phi_{n-1}(\reld)|.$$
If it holds that 
$|\phi_n(\reld)|\ne |\phi_i(\reld)|$ for every $i\in [n-1]$,
then we are done.
So we assume that 
there is an index $i \in [n-1]$ 
such that $|\phi_n(\reld)|= |\phi_i(\reld)|$.


Let $\reld'$ be the structure we get by applying Lemma~\ref{lem:twocasesemi} to $\phi_n$ and $\phi_i$.

Now choose $k$ such that for every $j$ with $1< j\le i$ we have \[\frac{|\phi_{j-1}(\reld)|^k}{|\phi_{j}(\reld)|^k} < \frac{1}{|\lib(\phi_{j-1})|^{|D'|}}.\]
Then we have for every $\ell\ge k$ and $1< j<i$
\begin{align*}
 |\phi_{j-1}(\reld^\ell \times \reld')| &= |\phi_{j-1}(\reld^\ell )| \cdot |\phi_{j-1}(\reld')|\\
 &\le |\phi_{j-1}(\reld^\ell )| \cdot |\lib(\phi_{j-1})|^{|D'|} \\
 &< |\phi_j(\reld^\ell)|\\
 &\le |\phi_j(\reld^\ell)| \cdot |\phi_j(\reld')|\\
 &= |\phi_j(\reld^\ell\times \reld')|.
\end{align*}
Analogously, we get for every $\ell>k$ that \[|\phi_{i-1}(\reld^\ell\times \reld')|<|\phi_n(\reld^\ell\times \reld')|.\]

Now choose $k'$ such that for every $j$ with $i \le j<n$ we have \[\frac{|\phi_{j+1}(\reld)|^{k'}}{|\phi_{j}(\reld)|^{k'}} > {|\lib(\phi_{j})|^{|D'|}}.\]
Then we have for every $\ell > k'$ and every $i \le j<n$
\begin{align*}
  |\phi_{j}(\reld^\ell \times \reld')| &= |\phi_{j}(\reld^\ell )| \cdot |\phi_{j}( \reld')|\\
 &\le |\phi_{j}(\reld^\ell )| \cdot |\lib(\phi_{j})|^{|D'|} \\
 &< |\phi_{j+1}(\reld^\ell)|\\
 &\le |\phi_{j+1}(\reld^\ell)| \cdot |\phi_j(\reld')|\\
 &= |\phi_{j+1}(\reld^\ell\times \reld')|.
\end{align*}
Similarly, we get for every $\ell >k$ that \[|\phi_{i+1}(\reld^\ell\times \reld')|>|\phi_n(\reld^\ell\times \reld')|.\]
Now choosing $\ell=\max(k,k')$ and noting that 
\begin{align*}|\phi_{i}(\reld^\ell\times \reld')|&= |\phi_{i}(\reld^\ell)|\cdot|\phi_{i}(\reld')|\\ & \ne |\phi_n(\reld^\ell)| \cdot |\phi_n(\reld')|  \\& = |\phi_n(\reld^\ell\times \reld')|
\end{align*}
completes the proof with $\relc= \reld^\ell\times \reld'$.
\end{proof}

The following is a consequence of this lemma.

\begin{lemma}
Let $\phi_1(S_1), \ldots, \phi_n(S_n)$ be connected,
liberal pp-formulas over the same vocabulary $\tau$
that are pairwise not counting equivalent.
Then there exists a structure $\relc$ (over $\tau$) such that
\begin{itemize}

\item for all pp-formulas $\phi$ (over $\tau$) it holds that
$|\phi(\relc)| > 0$, and

\item for all distinct $i, j \in [n]$, it holds that
$|\phi_i(\relc)| \neq |\phi_j(\relc)|$.

\end{itemize}
\end{lemma}

\begin{proof}
By Lemma~\ref{lem:semigeneral}, it suffices to show that
the pp-formulas $\phi_i$ are pairwise not semi-counting equivalent.
Since each $\phi_i$ is connected and liberal,
we have $\phi_i = \widehat{\phi_i}$.
Thus, by the hypothesis that the $\phi_i$ are pairwise not counting equivalent
in combination with
Theorem~\ref{thm:semiequichar},
we obtain that the $\phi_i$ are pairwise not semi-counting equivalent.
\end{proof}

\subsection{The all-free case}
\label{subsect:all-free}

The aim of this subsection is the proof of Theorem~\ref{thm:equivalence-theorem} in the special case of all-free ep-formulas.
Recall that an ep-formula is \emph{all-free} if
it is the disjunction of prenex pp-formulas, each of which
is \emph{free} in that it has a non-empty set of free variables. We will later in Section~\ref{subsect:general-equivalence-theorem} use the result on all-free formulas to prove the general version of Theorem~\ref{thm:equivalence-theorem}.

For every $\phi(V) \in \Phi$ we define a set $\phi^*$ of free pp-formulas;
then, we define $\Phi^* = \bigcup_{\phi \in \Phi} \phi^*(V)$.
Let $\phi(V)= \phi_1(V) \lor \ldots \lor \phi_s(V)$ where the 
$\phi_i(V)$ are free pp-formulas. By inclusion-exclusion we have for every structure $\relb$ that 
\begin{align}|\phi(\relb)| &= \sum_{J\in [s]} (-1)^{|J|+1}|(\bigwedge_{j\in J} \phi_j)(\relb)| \nonumber\\ & = \sum_{J\in [s]} (-1)^{|J|+1}|\phi_J(\relb)|,\end{align} 
where the $\phi_J(V) = \bigwedge_{j\in J} \phi_j(V)$ are pp-formulas. 
Now iteratively do the following: If there are two summands $c |\phi_J(\relb)|$ and $c' |\phi_{J'}(\relb)|$ such that $\phi_J$ and $\phi_{J'}$ are counting equivalent, delete both summands and add $(c+c')|\phi_J|$ to the sum. 
When this operation can no longer be applied, delete all summands with coefficient zero. The pp-formulas that remain in the sum form the set $\phi^*$. 

\begin{example}
\label{ex:phistar}
 It shall be advantageous to again consider Example~\ref{ex:eqsecond}. There we started off with \[\phi(V)= \phi_1(V) \lor \phi_2(V) \lor \phi_3(V).\] Inclusion-exclusion yields 
 \begin{align*}|\phi(\relb)| = &|\phi_1(\relb)| + |\phi_2(\relb)| + |\phi_3(\relb)|\\
  &- |(\phi_1\land \phi_2)(\relb)| - |(\phi_1\land \phi_3)(\relb)|  \\&- |(\phi_2 \land \phi_3)(\relb)|+ |(\phi_1 \land \phi_2 \land \phi_3)(\relb)|. \end{align*}
 Now we simplify as described above and get
 \[|\phi(\relb)| = 3 \cdot |\phi_1(\relb)| - 2\cdot |(\phi_1\land \phi_3)(\relb)| . \]
 Consequently, for this example we have \[\phi^* =\{\phi_1, \phi_1\land \phi_3\}.\]
\end{example}

The algorithm discussed above directly yields the following proposition.

\begin{proposition}\label{prop:sumwithcoeffs}
There exists an algorithm that, \\when
an all-free ep-formula $\phi$ is given as input, outputs
a set 
 $\phi^*:=\{\phi_1^*, \ldots, \phi_\ell^*\}$ of free pp-formulas,
 which are pairwise not counting equivalent,
 and coefficients $c_1, \ldots, c_\ell\in \mathbb{Z}\setminus\{0\}$ such that for every structure $\relb$,
$|\phi(\relb)|= \sum_{i=1}^\ell c_i |\phi^*_i(\relb)|$.
%
\end{proposition}

We will also require the following two facts for our proof.

\begin{proposition}\label{prop:nothomeq}
 Let us presume that $\phi(S)$ and $\phi'(S')$ are two semi-counting equivalent free pp-formulas that are not counting equivalent and let 
 $(\rela,S)$ and $(\rela',S')$ 
 be the structures of $\phi$ and $\phi'$, respectively. Then $\rela$ and $\rela'$ are not homomorphically equivalent.
\end{proposition}

\begin{proof}
$\phi(S)$ and $\phi'(S')$ 
are semi-counting equivalent, so we have by Theorem~\ref{thm:semiequichar} and Theorem~\ref{thm:renamingandcounting} that 
$\widehat{\phi(S)}$ and $\widehat{\phi'(S')}$ are renaming equivalent. 
It follows that $\rela$ and $\rela'$ are homomorphically
equivalent via homomorphisms
$h:A \to A'$, $h':A' \to A$
 that act as bijections between
$S$ and $S'$.

If there exists a homomorphism $g$ from $\rela$ to $\rela'$,
then we can extend $h$ (using the definition of $g$)
to be defined on the components of $\phi$ deleted
in the construction of $\widehat{\phi}$, to obtain a 
homomorphism from $\rela$ to $\rela'$ extending $h$.
If there exists a homomorphism $g'$ from $\rela'$ to $\rela$,
we can extend $h'$ in an analogous way.
However, the existence of both such extensions would 
imply by definition that $\phi(S)$ and $\phi'(S')$
are counting equivalent.
We may thus conclude that either 
there is no homomorphism $\rela\rightarrow \rela'$ or 
there is
no homomorphism $\rela'\rightarrow\rela$.
\end{proof}

\begin{lemma}\label{lem:splitsemieq}
There is an oracle FPT-algorithm that performs the following: 
given a set $\phi_1, \ldots, \phi_s$ of semi-counting equivalent free pp-formulas that are pairwise not counting equivalent, a sequence $c_1, \ldots, c_s\in \mathbb{Z}\setminus \{0\}$, and a structure $\relb$, 
the algorithm computes $|\phi_i(\relb)|$ for every $i\in [s]$;
it may make calls to an oracle that provides
$\sum_{i=1}^s c_i \cdot |\phi_i(\relb')|$ upon being given a structure $\relb'$.
Here, the $\phi_i$ with the $c_i$ constitute the parameter.
\end{lemma}

To establish this lemma, we first demonstrate the following proposition.

\begin{proposition}\label{prop:findCforsce}
 Let $\phi_1, \ldots, \phi_s$ be a sequence of semi-counting equivalent pp-formulas that are pairwise not counting equivalent. Then there is a structure $\relc$ and $i\in [s]$ such that $\relc\models \phi_i$ but $\relc\cancel{\models} \phi_j$ for all $j\in [s]\setminus \{i\}$.
\end{proposition}
\begin{proof}
Let $\rela_1, \ldots, \rela_n$ be the structures of the que\-ries $\phi_1, \ldots, \phi_n$. By Proposition~\ref{prop:nothomeq} the structures $\rela_i$ are pairwise not homomorphically equivalent. For $i,j\in [n]$, we write $\phi_i < \phi_j$ if there is a homomorphism from $\rela_i$ to $\rela_j$. It is easy to check that $<$ induces a partial order on the $\phi_i$. Let $\phi_i$ be a minimal element of this partial order, then there is no homomorphism from any $\rela_j$ to $\phi_i$ with $i\ne j$. Setting $\relc = \rela_i$ completes the proof.
\end{proof}

\begin{proof} (Lemma~\ref{lem:splitsemieq})
 We give and algorithm that recursively computes the $|\phi_i(\relb)|$ one after the other. So let the parameter and the input be given as in the statement of the lemma. By Proposition~\ref{prop:findCforsce}, there is an $i\in [n]$ and a structure $\relc$ such that $\relc\models \phi_i$ but $\relc\cancel{\models} \phi_j$ for all $j\in [s]\setminus \{i\}$. W.l.o.g.~assume $i=s$. Then $|\phi_i(\relb\times \relc)| = 0$ for $i<s$. Consequently, we have that the oracle lets us compute $c_s \cdot |\phi_n(\relb\times \relc)|=c_s\cdot|\phi_n(\relb)| \cdot |\phi_n(\relc)|$. Computing $|\phi_n(\relc)|$ by brute force then yields $|\phi_s(\relb)|$.
 
 Now note that for every structure $\relb'$ we can also compute $\sum_{i=1}^{s-1} c_i \cdot |\phi_i(\relb')|$ by this approach with one subtraction. So we can apply the algorithm again for $\phi_1, \ldots, \phi_{s-1}$, answering oracle queries for the smaller sum $\sum_{i=1}^{s-1} c_i \cdot |\phi_i(\relb')|$ with the help of the oracle for $\sum_{i=1}^s c_i \cdot |\phi_i(\relb')|$.
\end{proof}

We can now prove Theorem~\ref{thm:equivalence-theorem} for all-free ep-formulas.

\begin{theorem}\label{thm:reductionallfree}
Let $\Phi$ be a set of all-free ep-formulas.
There exists a set $\Phi^*$ of free prenex pp-formulas
such that the counting case problems
$\countp[\Phi]$ and $\countp[\Phi^*]$
are equivalent under counting slice reductions.
\end{theorem}

Before giving the technical details of the proof of Theorem~\ref{thm:reductionallfree}, let us first descibe the ideas. The proof follows the approach presented in the examples of Section~\ref{sct:examples}. In particular, the less straightforward reduction from $\countp[\Phi^*]$ to $\countp[\Phi]$ proceeds as we did in Example~\ref{ex:eqfirstrevisited}. Given $\phi$ and $\phi^*$, we can evaluate $|\phi'(\relb)|$ for $\phi'\in \phi^*$ with an oracle for $|\phi(\relb\times \relc^\ell)|$ for a suitable structure $\relc$ as in that example. The main difference is that, instead of having $\relc$ explicitly as in Example~\ref{ex:eqfirstrevisited}, we here know from Lemma~\ref{lem:semigeneral} that a structure $\relc$ exists for which all formulas in $\phi^*$ have a different number of satisfying assignments. We can then compute $\relc$ by brute force as it depends only on $\phi$. This then allows to compute $\phi'(\relb)$ by solving a system of linear equations.

We now give the technical detail of the proof.

\begin{proof}  
Let us first specify the reduction 
from $\countp[\Phi]$ to $\countp[\Phi^*]$,
which is quite straightforward.
 The relation $U$ is the set of pairs
 $(\phi, \phi^*)$
such that $\phi$ is an all-free ep-formula
and $\phi^*$ is the output of the algorithm of
Proposition~\ref{prop:sumwithcoeffs} on input $\phi$.
Obviously, this satisfies the coverage condition. Then the oracle-FPT-algorithm to compute $\phi(\relb)$ given $\phi, \phi^*$ and  $\relb$ first computes all of the $|\phi_i^*(\relb)|$ by oracle calls and then uses Proposition~\ref{prop:sumwithcoeffs}. This completes the reduction.

For the other direction, let $\phi'\in \Phi^*$. 
We set $U$
to be the set of all pairs
$(\phi', \{ \phi \})$
such that $\phi$ is an all-free ep-formula and $\phi' \in \phi^*$.
Given $\phi'$, $\phi$ and $\relb$, 
we compute $|\phi'(\relb)|:=r(\phi', \{ \phi \}, \relb)$ as follows: 
Let $\phi_1^*, \ldots, \phi_s^*$ be the equivalence classes of $\phi^*$ with respect to semi-counting equivalence. Now choose a strucuture $\relc$ as in Lemma~\ref{lem:semigeneral}. Then for $\psi, \psi'\in \phi^*$ we have $|\psi(\relc)|\ne |\psi'(\relc)|$ if $\psi$ and $\psi'$ are from different equivalence classes with respect to semi-counting equivalence, and otherwise $|\psi(\relc)|= |\psi'(\relc)| >0$. 
Fix for each $j\in [s]$ a formula in $\phi_j^*$ and call it $\psi_j$. Moreover, denote by $c_\psi$ the coefficiencent of $\psi$ in Proposition~\ref{prop:sumwithcoeffs}. Using this notation and Proposition~\ref{prop:sumwithcoeffs} we obtain, for every $\ell \in \nats$,
that
$$|\phi(\relb\times \relc^\ell)| 
= \sum_{j=1}^s |\psi_j(\relc)|^\ell(\sum_{\psi\in \phi^*_j} c_\psi |\psi(\relb)|).$$
Note that this is a linear equation where the coefficients 
have the form $|\psi_j(\relc)|^\ell$; 
these can be computed by brute force. Letting $\ell$ range from $0$ to $s-1$ thus yields a system of linear equations whose coefficient matrix is a Vandermonde matrix. Consequently, with $s$ oracle calls we can compute
 $\sum_{\psi\in \phi^*_j} c_\psi |\psi(\relb)|$ for each $j$. 
 We use Lemma~\ref{lem:splitsemieq} to compute $\phi'(\relb)$.
\end{proof}

\subsection{The general case}
\label{subsect:general-equivalence-theorem}

We now indicate how to prove Theorem~\ref{thm:equivalence-theorem}
in its full generality.

\newcommand{\af}{\mathrm{af}}

We may assume that each ep-formula $\phi \in \Phi$ is normalized.
For each ep-formula $\phi$, define $\phi_{\af}$ to be the
all-free part of $\phi$, that is, the
disjunction of the $\phi$-disjuncts that are free;
define $\Phi_{\af}$ to be $\{ \phi_{\af} ~|~ \phi \in \Phi \}$;
and,
define $\phi_{\af}^-$ to be 
the set of formulas in $\phi_{\af}^*$ 
that do not logically entail a sentence disjunct of $\phi$.
We define $\phi^+$ to be 
the union of $\phi_{\af}^-$ and
the set containing each pp-sentence disjunct of $\phi$;
and, we define $\Phi^+$ to be $\bigcup_{\phi \in \Phi} \phi^+$.


\begin{example}
Set $V = \{ w, x, y, z \}$; 
we consider the formulas 
$\phi(V) = \phi_1(V) \vee \phi_2(V) \vee \phi_3(V)$
defined in Example~\ref{ex:eqsecond}.
Define 
$$\theta_1(V) = \exists a \exists b \exists c \exists d
E(a,b) \wedge E(b,c) \wedge E(c,d),$$
and define
$$\theta(V) = \phi_1(V) \vee \phi_2(V) \vee \phi_3(V) \vee \theta_1(V).$$
The all-free part of $\theta$ is
$\theta_{\af} = \phi_1(V) \vee \phi_2(V) \vee \phi_3(V)$,
since each of these three disjuncts has a non-empty set of free variables,
whereas $\theta_1$ has an empty set of free variables.
According to Example~\ref{ex:phistar}, we have
$$\theta^*_{\af} = \{ \phi_1, \phi_1 \wedge \phi_3 \}.$$
Now, observe that $\phi_1 \wedge \phi_3$ logically entails
the sentence disjunct $\theta_1$ of $\theta$; on the other hand,
$\phi_1$ does not logically entail $\theta_1$.
Hence, we have that $\theta^-_{\af} = \{ \phi_1 \}$.
We have $\theta^+$ to be the union of $\theta^-_{\af}$ and
$\{ \theta_1 \}$, so
$$\theta^+ = \{ \phi_1, \theta_1 \}.$$
\end{example}

The idea of the proof of Theorem~\ref{thm:equivalence-theorem}
is as follows.
The counting slice reduction
from $\countp[\Phi]$ to $\countp[\Phi^+]$ has
$U$ as the set of pairs $(\phi, \phi^+)$ where $\phi$
is a normalized ep-formula; $r$ on $(\phi(V), \phi^+, \relb)$
behaves as follows.
First, it checks if there is a sentence disjunct $\theta$ of $\phi$
that is true on $\relb$; if so, it outputs $|B|^{|V|}$;
otherwise, it makes use of the counting slice reduction 
from 
$\countp[\Phi_{\af}]$ to $\countp[\Phi_{\af}^*]$.
The counting slice reduction from
$\countp[\Phi^+]$ to $\countp[\Phi]$
has $U$ as the set $\{(\psi, \{\phi\}) ~|~ \psi \in \phi^+ \}$;
$r$ on $(\psi(V), \phi(V), \relb)$ is defined as follows.
When $\psi \in \phi^-_{\af}$, the counting slice reduction $(U', r')$
from
$\countp[\Phi^*_{\af}]$ to $\countp[\Phi_{\af}]$
is used to determine $|\psi(\relb)|$; this is performed
by passing to $r'$ a treated version of $\relb$,
on which no sentence disjunct of $\phi$ may hold.
When $\psi$ is a sentence disjunct of $\phi$, 
an oracle query is made
to obtain the count of $\phi$ on a treated version of $\relb$;
on this treated version, it is proved that
all assignments satisfy $\phi$ 
if and only if $\relb \models \psi$.

\section{Conclusion}


We have shown a trichotomy for the parameterized complexity of counting satisfying assignments to  existential positive formulas of bounded arity. To this end, the main technical contribution was the equivalence theorem (Theorem~\ref{thm:equivalence-theorem}) stating that for every set of existential positive formulas there is a set of  primitive positive formulas that is computationally equivalent with respect to 
the counting problem studied. After showing this equivalence theorem, we could derive our trichotomy in a rather straightforward fashion by 
invoking a previous trichotomy for primitive positive formulas as a black-box.

In order to prove the equivalence theorem, we gave a syntactic characterization for when two pp-formulas are counting equivalent, that is, have the same number of satisfying assignments with respect to every finite structure. 
This result can be seen as an adaption, to the counting setting,
 of classical work of Chandra and Merlin~\cite{ChandraMerlin77-optimal} that characterizes logical equivalence of primitive positive formulas. 

Let us note that the assumption of bounded arity is not needed in the proof of the equivalence theorem. It only appears in our trichotomy theorem because it is already present in the previous trichotomy on primitive positive formulas that we use. Consequently, if one could adapt the work of 
Marx~\cite{Marx10-tractablehypergraph}  
on model checking unbounded arity primitive positive formulas to counting to show a dichotomy or trichotomy for counting, this would directly give the corresponding result for existential positive formulas by applying our equivalence theorem.

Finally, let us remark that we are not aware of any 
fragment of first-order logic extending existential positive queries
for which even model checking is understood, from the viewpoint of classifying the complexity of all sets of queries
(for more information, see the discussion in
the introduction of the article~\cite{Chen14-frontier}).
Hence,
the research project of extending our 
complexity classification beyond existential positive queries
would first require an advance in the study of model checking
in first-order logic.



\newpage

\arxivconf{\bibliographystyle{alpha}}{\bibliographystyle{plain}}
\bibliography{../../hubiebib}

\begin{thebibliography}{Mar10b}

\bibitem[AHV95]{AbiteboulHullVianu95-foundationsdatabases}
S.~Abiteboul, R.~Hull, and V.~Vianu.
\newblock {\em Foundations of Databases}.
\newblock Addison-Wesley, 1995.

\bibitem[CG10]{ChenGrohe10-succinct}
Hubie Chen and Martin Grohe.
\newblock Constraint satisfaction with succinctly specified relations.
\newblock {\em Journal of Computer and System Sciences}, 76(8):847--860, 2010.

\bibitem[Che14a]{Chen14-existentialpositive}
Hubie Chen.
\newblock On the complexity of existential positive queries.
\newblock {\em ACM Trans. Comput. Log.}, 15(1), 2014.

\bibitem[Che14b]{Chen14-frontier}
Hubie Chen.
\newblock The tractability frontier of graph-like first-order query sets.
\newblock In {\em Joint Meeting of the Twenty-Third {EACSL} Annual Conference
  on Computer Science Logic {(CSL)} and the Twenty-Ninth Annual {ACM/IEEE}
  Symposium on Logic in Computer Science (LICS), {CSL-LICS} '14, Vienna,
  Austria, July 14 - 18, 2014}, page~31, 2014.

\bibitem[CM77]{ChandraMerlin77-optimal}
Ashok~K. Chandra and Philip~M. Merlin.
\newblock Optimal implementation of conjunctive queries in relational data
  bases.
\newblock In {\em Proceddings of {STOC}'77}, pages 77--90, 1977.

\bibitem[CM14a]{ChenMengel14-pp-arxiv}
Hubie Chen and Stefan Mengel.
\newblock A trichotomy in the complexity of counting answers to conjunctive
  queries.
\newblock {\em CoRR}, abs/1408.0890, 2014.

\bibitem[CM14b]{ChenMueller14-hierarchy}
Hubie Chen and Moritz M{\"{u}}ller.
\newblock One hierarchy spawns another: graph deconstructions and the
  complexity classification of conjunctive queries.
\newblock In {\em Joint Meeting of the Twenty-Third {EACSL} Annual Conference
  on Computer Science Logic {(CSL)} and the Twenty-Ninth Annual {ACM/IEEE}
  Symposium on Logic in Computer Science (LICS), {CSL-LICS} '14, Vienna,
  Austria, July 14 - 18, 2014}, pages 32:1--32:10, 2014.

\bibitem[CM15]{ChenMengel15-pp-icdt}
Hubie Chen and Stefan Mengel.
\newblock A trichotomy in the complexity of counting answers to conjunctive
  queries.
\newblock In {\em 18th International Conference on Database Theory, {ICDT}
  2015, March 23-27, 2015, Brussels, Belgium}, pages 110--126, 2015.

\bibitem[DJ04]{DalmauJonsson04-counting}
V{\'{\i}}ctor Dalmau and Peter Jonsson.
\newblock The complexity of counting homomorphisms seen from the other side.
\newblock {\em Theor. Comput. Sci.}, 329(1-3):315--323, 2004.

\bibitem[DM13]{DurandMengel13-structuralcounting}
A.~Durand and S.~Mengel.
\newblock Structural tractability of counting of solutions to conjunctive
  queries.
\newblock In {\em Proceedings of the 16th International Conference on Database
  Theory (ICDT 2013)}, 2013.

\bibitem[FG06]{FlumGrohe06-parameterizedcomplexity}
J.~Flum and M.~Grohe.
\newblock {\em Parameterized Complexity Theory}.
\newblock Springer, 2006.

\bibitem[Gro07]{Grohe07-otherside}
Martin Grohe.
\newblock The complexity of homomorphism and constraint satisfaction problems
  seen from the other side.
\newblock {\em Journal of the ACM}, 54(1), 2007.

\bibitem[GS14]{GrecoScarcello14-counting}
Gianluigi Greco and Francesco Scarcello.
\newblock Counting solutions to conjunctive queries: structural and hybrid
  tractability.
\newblock In {\em Proceedings of the 33rd {ACM} {SIGMOD-SIGACT-SIGART}
  Symposium on Principles of Database Systems}, pages 132--143, 2014.

\bibitem[GSS01]{GroheSchwentickSegoufin01-conjunctivequeries}
Martin Grohe, Thomas Schwentick, and Luc Segoufin.
\newblock When is the evaluation of conjunctive queries tractable?
\newblock In {\em STOC 2001}, 2001.

\bibitem[Mar10a]{Marx10-canyoubeat}
D{\'{a}}niel Marx.
\newblock Can you beat treewidth?
\newblock {\em Theory of Computing}, 6(1):85--112, 2010.

\bibitem[Mar10b]{Marx10-tractablehypergraph}
D{\'a}niel Marx.
\newblock Tractable hypergraph properties for constraint satisfaction and
  conjunctive queries.
\newblock In {\em Proceedings of the 42nd ACM Symposium on Theory of
  Computing}, pages 735--744, 2010.

\bibitem[PS11]{PichlerSkritek11-counting}
Reinhard Pichler and Sebastian Skritek.
\newblock Tractable counting of the answers to conjunctive queries.
\newblock In {\em Proceedings of the 5th Alberto Mendelzon International
  Workshop on Foundations of Data Management}, 2011.

\bibitem[PY99]{PapadimitriouYannakakis99-database}
C.~Papadimitriou and M.~Yannakakis.
\newblock {On the Complexity of Database Queries}.
\newblock {\em Journal of Computer and System Sciences}, 58(3):407--427, 1999.

\end{thebibliography}

\newpage
\appendix






\section{Proof of Theorem 3.1}

\label{subsect:proof-of-equiv-theorem}
\begin{proof} (Theorem~\ref{thm:equivalence-theorem})
We first describe a counting slice reduction $(U, r)$ from 
$\countp[\Phi]$ to $\countp[\Phi^+]$.
Let $(U', r')$ be the counting slice reduction from
$\countp[\Phi_{\af}]$ to $\countp[\Phi_{\af}^*]$
given by Theorem~\ref{thm:reductionallfree}.
Define $U$ to be the set 
$\{ (\phi, \phi^+) ~|~ \textup{$\phi$ is a normalized ep-formula } \}$.
When $(\phi, \phi^+) \in U$,
we define
$r(\phi(V), \phi^+, \relb)$ to be the result of the following
algorithm, which is FPT with respect to $(\pi_1, \pi_2)$.
For each sentence disjunct $\theta$ of $\phi(V)$, 
the algorithm queries 
$\countp(\theta, \relb)$; if for some such disjunct $\theta$
it holds that $\relb \models \theta$, then the algorithm outputs $|V|^{|B|}$.
Otherwise, for any assignment $f: V \to B$,
it holds that $\relb, f \models \phi$ 
if and only if $\relb, f \models \phi_{\af}$.
So, the algorithm returns 
$r'(\phi_{\af}, \phi_{\af}^*, \relb)$ by running 
the corresponding algorithm for $r'$.
In this run, the algorithm for $r'$ only makes queries of the form
$(\psi, \relb)$ (with $\psi \in \phi_{\af}^*$);
such queries where $\psi \in \phi^{-}_{\af}$ are resolved using
the oracle in the definition of counting slice reduction,
and queries where $\psi \in \phi_{\af}^* \setminus \phi_{\af}^-$
are answered with $0$.
Correctness is straightforward to verify.

We next describe a counting slice reduction 
$(U, r)$ from
$\countp[\Phi^+]$ to $\countp[\Phi]$.
Let $(U', r')$ denote the counting slice reduction from
$\countp[\Phi^*_{\af}]$ to $\countp[\Phi_{\af}]$
given by Theorem~\ref{thm:reductionallfree}.
We define $U:=\{ (\psi, \{ \phi \}) ~|~ \psi \in \phi^+ \}$.
We need to define
$r(\psi(V), \phi(V), \relb)$ when $(\psi, \phi) \in U$.

Let us describe first an algorithm for \arxivversion{the mapping} $r$
in the case that $\psi \in \phi^{-}_{\af}$.
Let $(\relc_1, V), \ldots, (\relc_m, V)$ denote the 
pp-formulas in $\phi^{-}_{\af}$, and let 
$\relc$ denote the disjoint union of the structures $\relc_i$.
Observe that for any structure $\reld$, it holds that
$\reld \times \relc, f \models \phi$
if and only if $\reld \times \relc, f \models \phi_{\af}$, 
since no sentence disjunct of $\phi$ holds on $\relc$
(due to the definitions of $\relc$ and $\phi^{-}_{\af}$).
Call the algorithm for $r'$ to compute
$r'(\psi, \{ \phi_{\af} \}, \relb \times \relc) = 
|\psi(\relb \times \relc)|$;
note that the oracle queries made by this algorithm
can be resolved by an oracle for $\countp(\phi,\cdot)$, 
since all such oracle queries have the form 
$\countp(\phi_{\af}, \cdot \times  \relc)$.
As $|\psi(\relb \times \relc)| = |\psi(\relb)| \cdot |\psi(\relc)|$,
by dividing this quantity by $|\psi(\relc)|$, 
one can determine $|\psi(\relb)|$, which is the desired value.
Note that by the definition of $\relc$, it holds that 
$|\psi(\relc)|$ is non-zero.

In order to describe the behavior of the algorithm for $r$
in the case that $\psi$ is a sentence disjunct of $\phi$,
we establish the following claim.
Let $(\rela, V)$ be the structure view of $\psi$.	

{\bf Claim:} Let $i: V \to V$ be the identity map on $V$.
For each disjunct $\theta$ of $\phi$, it holds that
$\rela, i \models \theta(V)$ if and only if $\theta = \psi$.

The backwards direction is clear, so we prove the forwards direction.
If a disjunct $\theta$ is a free pp-formula, then 
$\rela, i \not\models \theta(V)$ since
$\theta$ contains an atom using a variable $v \in V$,
whereas no tuple of a relation of $\rela$ contains any variable from $V$.
If a disjunct $\theta$ is a pp-sentence $(\rela', V)$
not equal to $\psi$, then by definition of \emph{normalized} ep-formula,
there is no homomorphism from $\rela'$ to $\rela$ and hence
$\rela, i \not\models \theta(V)$.  This establishes the claim.

Now suppose that $\psi$ is a sentence disjunct of $\phi$.
In this case, the algorithm for $r(\psi(V), \phi(V), \relb)$ 
behaves as follows.
It queries $\countp(\phi, \rela \times \relb)$ to determine
$|\phi(\rela \times \relb)|$;
it outputs $|B|^{|V|}$ if $|\phi(\rela \times \relb)|$ is equal
to $(|A|\cdot|B|)^{|V|}$ (the \emph{maximum count} possible there),
and outputs $0$ otherwise.
We prove that this is correct by 
showing that $|\phi(\rela \times \relb)|$ is the maximum count
if and only if $\relb \models \psi$.

For the backwards direction, suppose that $\relb \models \psi$,
and denote $\psi$ by $(\rela, V)$.  Then, there is a homomorphism
from $\rela$ to $\relb$, and hence there is a homomorphism
from $\rela$ to $\rela \times \relb$.
It follows that for any assignment $f: V \to V$,
one has $\rela \times \relb, f \models \psi(V)$.
For the forwards direction, suppose that
$|\phi(\rela \times \relb)|$ is the maximum count.
Let $i': V \to A \times B$ be any map
such that for each $v \in V$, the value
$i'(v)$ has the form $(i(v),j(v))$
where $j: V \to B$ is a map.
We have that $\rela \times \relb, i' \not\models \phi(V)$.
It follows that there is a disjunct $\theta$ of $\phi$
such that
$\rela \times \relb, i' \models \theta(V)$.
It follows that
$\rela, i \models \theta(V)$ and $\relb, j \models \theta(V)$.
By the claim established above, we have that $\theta = \psi$.
Then, it holds that $\relb, j \models \psi$, and we are done.
\end{proof}









%




\end{document}